\documentclass[aps, preprint]{revtex4-2}

\usepackage[utf8]{inputenc}
\usepackage[T1]{fontenc}
\usepackage{amsmath}
\usepackage{amsthm}
\usepackage{amssymb}
\usepackage{hyperref}
\usepackage{dsfont}
\usepackage{graphicx}
\usepackage{enumerate}

\hypersetup{colorlinks=true, linkcolor=blue, citecolor=blue}

\newcommand{\Real}{\mathbb{R}}

\newcommand{\Comp}{\mathbb{C}}
\newcommand{\seq}{\subseteq}
\newcommand{\pspace}{\Real^3 \setminus \{0\}}
\newcommand{\LittleGroup}[1]{H_{#1}}
\newcommand{\Lightcone}{\mathcal{L}_+}
\newcommand{\poincare}{\mathrm{ISO}^+(3,1)}
\newcommand{\Poincare}{Poincar\'{e} }
\newcommand{\SO}{\mathrm{SO}}
\newcommand{\so}{\mathfrak{so}}

\newcommand{\ISO}{\mathrm{ISO}}
\newcommand{\bfk}{\boldsymbol{k}}
\newcommand{\diff}{\mathrm{Diff}}

\newtheorem{theorem}{Theorem}
\newtheorem{corollary}{Corollary}
\newtheorem{lemma}[theorem]{Lemma}
\newtheorem{noGoTheorem}{No-Go Theorem}

\begin{document}

\title{Four no-go theorems on the existence of spin and orbital angular momentum of massless bosons}
\date{\today}

\author{Eric Palmerduca}
\email{ep11@princeton.edu}
\affiliation{Department of Astrophysical Sciences, Princeton University, Princeton, New Jersey 08544}
\affiliation{Plasma Physics Laboratory, Princeton University, Princeton, NJ 08543,
U.S.A}

\author{Hong Qin}
\email{hongqin@princeton.edu}
\affiliation{Department of Astrophysical Sciences, Princeton University, Princeton, New Jersey 08544}
\affiliation{Plasma Physics Laboratory, Princeton University, Princeton, NJ 08543,
U.S.A}

\begin{abstract}
The past decades have seen substantial interest in the so-called orbital angular momentum (OAM) of light, driven largely by its diverse range of applications. However, there are fundamental theoretical issues with decomposing the angular momentum of massless particles, such as photons, into spin (SAM) and orbital angular momentum parts. While the angular momentum of massive particles has a natural splitting into the Wigner SAM and OAM, there are numerous proposed splittings for photons and no consensus about which is correct. Moreover, it has been shown that most of the proposed SAM and OAM operators do not satisfy the defining commutation relations of angular momentum operators and are thus not legitimate splittings. Here, we prove that it is generally impossible to split the total angular momentum operator of massless bosons, such as photons and gravitons, into spin and orbital parts. We prove two further generalizations of this result, showing that there are no SAM-OAM splittings even if (1) the SAM operator generates non-internal symmetries or (2) if one allows the SAM and OAM operators to generate non-$\SO(3)$ symmetries.
\end{abstract}

\maketitle

\section{Introduction}
Since it was first introduced, the concept of the orbital angular momentum (OAM) of light has been a popular area of research \cite{Shen2019}. It has found diverse scientific and industrial applications including, but not limited to, imaging \cite{Furhapter2005,Tamburini2006}, optical communication \cite{Barreiro2008}, and the manipulation of microstructures \cite{Grier2003}. Following the detection of gravitational waves by LIGO \cite{GravWave2016Blackhole,GravWave2017NeutronStar}, researchers have also been interested in the OAM of gravitational waves \cite{Bialynicki-Birula_2016, Baral2020, Wu2023}. However, there are fundamental theoretical issues with splitting the angular momentum of massless bosons, such as photons and gravitons, into spin angular momentum (SAM) and OAM \cite{VanEnk1994,FernandezCorbaton2014,PalmerducaQin_PT,PalmerducaQin_GT}. Angular momentum operators, such as SAM and OAM operators, must be generators of three-dimensional rotational symmetries of the particle, and must therefore satisfy $\so(3)$ commutation relations \cite{VanEnk1994,FernandezCorbaton2014,Leader2019, PalmerducaQin_PT, PalmerducaQin_GT}. Furthermore, spin is associated with the internal degrees of freedom (DOFs) \cite{Shankar_QM}, so SAM operators must generate an $\SO(3)$ symmetry of the internal space of a particle. For massive particles, the angular momentum naturally splits into SAM and OAM. In contrast, there is no agreed upon method to define these operators for massless bosons, leading to a vast array of different proposed SAM and OAM operators \cite{Akhiezer1965,Jaffe1990,VanEnk1994,Chen2008,Wakamatsu2010,Bialynicki-Birula2011,Leader2013,Leader2014,Leader2016,Leader2019,Yang2022}. However, most of these operators have been shown to either not be well-defined or else not actually satisfy $\so(3)$ commutation relations, and are therefore not genuine angular momentum operators \cite{VanEnk1994,FernandezCorbaton2014,Leader2019,Yang2022,PalmerducaQin_PT,PalmerducaQin_GT}. We recently showed \cite{PalmerducaQin_PT} that in some of the proposed SAM-OAM splittings for photons \cite{VanEnk1994,Bialynicki-Birula2011}, the SAM operators actually generate an $\Real^3$ symmetry rather than an $\SO(3)$ symmetry, while the OAM operators do not generate any symmetry at all. In the case of gravitons, we proved the stronger result that there is no possible SAM-OAM splitting of the angular momentum \cite{PalmerducaQin_GT}. In this article, we generalize and extend these arguments, showing that the same is true for any massless particle. We do so by proving four no-go theorems. The first theorem uses geometric methods to show that the angular momentum operator of massless particles cannot split into nontrivial SAM and OAM operators. The third and fourth theorems show that various generalizations of the definition of SAM and OAM operators do not remove the obstruction. More specifically, the third no-go theorem shows that it is not possible to split the angular momentum operators of massless particles into generators of non-$\SO(3)$ symmetries, which, for example, was attempted in Refs. \cite{VanEnk1994,Bialynicki-Birula2011}. The fourth no-go theorem shows the stronger result that it is not possible to split the angular momentum operator of massless particles into two generators of any nontrivial energy-preserving symmetries, even if neither of those symmetries are internal.

While the proofs of the first, third, and fourth theorems are essentially geometric, we show in the second no-go theorem that the obstruction preventing the splitting of angular momentum can also be viewed as a consequence of the nontrivial topology of massless particles. This nontrivial topology reflects the intertwining of the internal and external DOFs which prevents an SAM-OAM decomposition. We also show how gauge redundant descriptions of massless particles can give the false impression that the angular momentum splits into SAM and OAM parts.

This article is organized as follows. In Sec. \ref{sec:Formalism}, we review the vector bundle formalism for elementary particles and discuss its equivalence with the usual Hilbert space representations. In Secs. \ref{sec:no_go_1}, \ref{sec:no_go_2}, \ref{sec:no_go_3},  and \ref{sec:no_go_4} we prove No-Go Theorems \ref{thm:nogo1}, \ref{thm:nogo2}, \ref{thm:nogo3}, and \ref{thm:nogo4}, respectively. In Sec. \ref{sec:gauge_redundancy}, we show how gauge redundant descriptions can lead to non-physical SAM-OAM splittings for massless particles. In Sec. \ref{sec:Fock_space} we discuss the validity of using single-particle states rather than multiparticle Fock states to prove the no-go theorems.

\section{Formalism}\label{sec:Formalism}
Conventionally, elementary particles are defined as unitary irreducible representations (UIRs) of the proper orthochronous \Poincare group $\poincare$ on a Hilbert space \cite{Wigner1939,Weinberg1995}. Equivalently, they can be considered as UIRs on a Hermitian vector bundle $\pi:E \rightarrow M$ where $E$ is the total space, $M$ is the momentum space of the particle, and $\pi$ is the projection $(k,v) \mapsto k$ \cite{Simms1968, Asorey1985, PalmerducaQin_PT, PalmerducaQin_GT}. It is typical to also refer to the total space $E$ as the vector bundle. We also refer to the vector $(k,v)$ by just $v$ when the base point $k$ is irrelevant. Let $\mathcal{B}$ be Minkowski space with signature $(-,+,+,+)$. For massive particles with $m>0$, the momentum space $M$ is the (positive) mass hyperboloid
\begin{equation}
    \mathcal{M}_m = \{k^\mu = (\omega, \boldsymbol{k}) \in \mathcal{B}: k^\mu k_\mu= -m^2\} \cong \mathbb{R}^3
\end{equation}
while for massless particles $M$ is the forward lightcone
\begin{align}
    \Lightcone &= \{k^\mu= (\omega, \boldsymbol{k}) \in \mathcal{B}: k^\mu k_\mu = 0, \omega > 0 \} \\
    &\cong \pspace.
\end{align}
Note that massless particles have no rest frame and thus $\omega =|\boldsymbol{k}|= 0$ is excluded from the lightcone. Parametrizing $\mathcal{M}_m$ and $\Lightcone$ by just their spatial parts $\boldsymbol{k}$ gives the diffeomorphisms with $\Real^3$ and $\pspace$, respectively.
The fiber at $k \in M$ is the vector space $E(k) \doteq \pi^{-1}(k)$. The momentum space $M$ describes the external DOFs of the particle while the fiber DOFs describe the internal state (i.e., $v\in E(k)$ describes the polarization). The vector bundle $E$ is Hermitian, meaning that for each $k$, there is a Hermitian inner product $\langle \cdot , \cdot \rangle$ on the fiber $E(k)$. That $E$ is a representation of $\poincare$ means that there are $\poincare$-actions $\Sigma$ and $\sigma$ on $E$ and $M$, respectively, such that for $\Lambda \in \poincare$, $\Sigma_\Lambda$ is a vector space isomorphism of $E(k)$ onto $E(\sigma_\Lambda k)$. Furthermore, $\Sigma$ and $\sigma$ respect the group structure of $\poincare$, that is,
\begin{align}
    \Sigma_{\Lambda_1 \Lambda _2} = \Sigma_{\Lambda_1}\Sigma_{\Lambda_2} \\
    \Sigma_{\Lambda ^{-1}} = (\Sigma_{\Lambda})^{-1},
\end{align}
and similarly for $\sigma$.
We will mostly consider bundles for which $\sigma$ is the standard action of $\poincare$ on momentum space, that is, when $\Lambda \in \SO^+(3,1)$ is a Lorentz transformation we simply have $\sigma_\Lambda k = \Lambda k$, while the inhomogeneous spacetime translations in $\poincare$ do not effect $k$. That the action $\Sigma$ is unitary means that
\begin{equation}
    \langle \Sigma_{\Lambda} v_1 , \Sigma_{\Lambda} v_2 \rangle = \langle v_1, v_2 \rangle
\end{equation}
for all $\Lambda \in \poincare$ and all $v_1, v_2 \in E(k)$. The irreducibility of the action means that there are no proper subbundles $E'$ (with nonzero rank) of $E$ which are preserved by $\Sigma$. These definitions generalize the notion of vector space representations to vector bundle representations. One can recover the conventional Hilbert space representations $\tilde{\Sigma}$ of the particles by considering the induced UIR on the Hilbert space $L^2(E)$ of square-integrable sections of $E$ \cite{Simms1968,PalmerducaQin_PT}. In particular, for a wavefunction $\psi(k) \in L^2(E)$, the action is
\begin{equation}
    [\tilde{\Sigma}_{\Lambda} \psi](k) \doteq \Sigma_\Lambda [\psi(\sigma_{\Lambda^{-1}} k)].
\end{equation}
The inner product of $E$ induces an inner product on $L^2(E)$ given by
\begin{equation}
    \langle \psi_1 , \psi_2 \rangle = \int_M \langle \psi_1(k), \psi_2(k) \rangle d\xi
\end{equation}
where $d\xi$ is the Lorentz invariant measure on $M$; if $M = \Lightcone$, then $d\xi = |\boldsymbol{k}|^{-1} d^3\boldsymbol{k}$. Under this product, $\tilde{\Sigma}$ is a UIR of $\poincare$ on the Hilbert space $L^2(E)$. The Lie group representations $\Sigma$ and $\tilde{\Sigma}$ of $\poincare$ induce  corresponding Lie algebra representations $\eta$  and $\tilde{\eta}$ of $\mathfrak{iso}(3,1)$ describing the infinitesimal group actions. For $\mathfrak{g} \in \mathfrak{iso}(3,1)$, $\eta_\mathfrak{g}$ is the tangent vector field on $E$ given by
\begin{equation}\label{eq:eta_g_bundle}
    \eta_\mathfrak{g}(k,v) \doteq \frac{d}{dt}\Big|_{t=0}\Sigma_{\exp(i t\mathfrak{g})}(k,v) \in T_{(k,v)}E
\end{equation}
and $\tilde{\eta}_\mathfrak{g}$ is an operator on $L^2(E)$ given by
\begin{equation}
    \tilde{\eta}_\mathfrak{g} \psi = \frac{d}{dt}\Big|_{t=0}\tilde{\Sigma}_{\exp(i t\mathfrak{g})}\psi.
\end{equation}
Note that $\eta_{\mathfrak{g}}$ is an operator in the sense that vector fields on $E$ are differential operators on the space of smooth $\Comp$-valued functions on $E$. The (total) angular momentum operators of these representations, $\boldsymbol{J} = (J_1,J_2,J_3)$ and $\boldsymbol{\tilde{J}} = (\tilde{J}_1, \tilde{J}_2, \tilde{J}_3)$, are the generators of rotation, corresponding to $\eta$ and $\tilde{\eta}$ when the Lie algebra $\mathfrak{iso}(3,1)$ is restricted to $\so(3)$. They satisfy the $\so(3)$ commutation relations 
\begin{align}
    [J_m,J_n] = i\epsilon_{mnp}J_p \label{eq:J_bundle_comm}\\ 
    [\tilde{J}_m, \tilde{J}_n] = i\epsilon_{mnp}\tilde{J}_p \label{eq:J_Hilbert_comm}
\end{align}
where the bracket in the first equation is the Jacobi-Lie bracket and that in the second equation is the usual operator commutator.\footnote{Note that we use the physics convention for the Lie algebra, in which $\so(3)$ corresponds to unitary transformations. Mathematicians use a different convention in which $\so(3)$ corresponds to anti-unitary transformations. In that convention, the factors of $i$ would not appear in equations (\ref{eq:eta_g_bundle})-(\ref{eq:J_Hilbert_comm}) For a short discussion, see Section 3.4 Ref. \cite{Hall2015}.} We restrict our discussion to non-projective representations of $\poincare$, which corresponds to treating only elementary particles with integer spin or helicity, i.e., bosons. All  elementary massless particles in the Standard Model are bosons, so this assumption has little effect on the generality of our results. We note also that a three-component operator $\boldsymbol{V}$ which satisfies 
\begin{equation}\label{eq:vector_operator_condtion}
    [V_m,J_n] = i\epsilon_{mnp}V_p
\end{equation}
transforms under rotations like a spatial vector, and is thus called a vector operator \cite{Hall2013, Hall2015}.

The bundle action given by $\Sigma$ and $\boldsymbol{J}$ describes the action on single-particle states while $\tilde \Sigma$ and $\tilde{\boldsymbol{J}}$ describe the action on wave functions; they contain equivalent information. While the action on wave functions is more commonly used, we will mostly work with the bundle action as it leads to simpler proofs of the no-go theorems. The reason is that these proofs make frequent use of  sharp momentum states, that is, states which have a single momentum $p \in M$. These are simply the vectors in $E$ in the bundle representation, while in the wave function representation these are $\delta$-function-like states. However, Flato et al. \cite{Flato1983} showed that for massless particles these sharp states are not standard delta functions, but rather twisted delta (``twelta'') functions, which require more care to work with. We can avoid these complications by working in the bundle representation.

\section{Massless particles do not have SAM or OAM: A geometric proof}\label{sec:no_go_1}

The aim of this article is to study possible decompositions of the total angular momentum $\boldsymbol{J}$ into nonzero vector operators $\boldsymbol{S}$ and $\boldsymbol{L}$:
\begin{equation}\label{eq:canon_split}
    \boldsymbol{J} = \boldsymbol{S} + \boldsymbol{L}.
\end{equation}
In this article, we take the defining property of angular momentum operators to be that they generate $\SO(3)$ symmetries, as is standard in the literature \cite{VanEnk1994,FernandezCorbaton2014,Hall2013,Leader2014,Leader2019,PalmerducaQin_PT,PalmerducaQin_GT}. This ensures that standard results about angular momentum operators hold such as the addition of angular momentum equation, the Clebsch-Gordan formalism, the Casmimir invariance of the $S^2$ and $L^2$, and the angular momentum multiplet structure of the eigenspaces of $(L^2,L_z)$ and $(S^2,S_z)$. We note that there has been some work in which this requirement on angular momentum operators is not imposed \cite{Yang2022,Das2024} and that our results do not apply in such a generalized setting. If $\boldsymbol{S}$ and $\boldsymbol{L}$ are to be considered angular momentum operators under the standard definition, they must (up to Lie algebra isomorphism) satisfy \cite{VanEnk1994,Hall2013,Yang2022,PalmerducaQin_PT,PalmerducaQin_GT}
\begin{align}
    [S_m,S_n] &= i\epsilon_{mnp} S_p \label{eq:S_vect_op}\\
    [L_m,L_n] &= i\epsilon_{mnp} L_p. \label{eq:L_vect_op}
\end{align}
These generate corresponding bundle actions $(\Sigma^S,\sigma^S)$ and $(\Sigma^L, \sigma^L)$ of $\SO(3)$ on $\pi:E\rightarrow M$. Since SAM describes an $\SO(3)$ symmetry of the internal (fiber) DOFs, $\Sigma^S$ must not change the momentum, i.e., $\sigma^S_R(\boldsymbol{k}) = \boldsymbol{k}$ for every $R \in \SO(3)$. Such actions which do not change the fiber are called stabilizing. Under such a $\Sigma^S$, each fixed fiber $E(k)$ is a finite-dimensional vector space representation of $\SO(3)$.

We now illustrate that there is a natural splitting of the form (\ref{eq:canon_split}) for massive particles following from Wigner's little group method \cite{Wigner1939,Weinberg1995,Terno2003}, but that this construction breaks down for massless particles. For an action $\Sigma$ of $\poincare$ on $\pi:E\rightarrow M$ with $\sigma_{\Lambda} = \Lambda$, the little group at $k \in M$ is the subgroup $H_k$ of $\SO^+(3,1)$ consisting of elements which leave $k$ invariant \cite{Wigner1939,Weinberg1995}:
\begin{equation}
    \LittleGroup{k} \doteq \{\Lambda\in \mathrm{SO}^+(3,1)|\Lambda k= k\}.
\end{equation}
Up to group isomorphism, $\LittleGroup{k}$ is independent of $k$, so we can just refer to the little group $H$ \cite{Weinberg1995}. There is a canonically induced unitary action of $H$ on $E$ which is stabilizing \cite{PalmerducaQin_PT}. For a particle $\pi_m:E_m \rightarrow \mathcal{M}_m$ with mass $m>0$, the little group is $\SO(3)$. Indeed, at $\boldsymbol{k} = \boldsymbol{0}$, we have that $H_{\boldsymbol{0}} = \SO(3) \seq \ISO^+(3,1)$ is the canonical copy of $\SO(3)$ in the \Poincare group consisting of pure spatial rotations. At any other $\boldsymbol{k}$, let $\Lambda_{\boldsymbol{k}}$ be the unique boost taking $(m,\boldsymbol{0})$ to $(\sqrt{m^2 + |\boldsymbol{k}|^2},\boldsymbol{k})$. At $\boldsymbol{k}$, the little group is
\begin{equation}
    H_{\boldsymbol{k}} = \Lambda_{\boldsymbol{k}} H_{\boldsymbol{0}} \Lambda_{-\boldsymbol{k}} = \Lambda_{\boldsymbol{k}} H_{\boldsymbol{0}} \Lambda_{\boldsymbol{k}}^{-1} \cong \SO(3).
\end{equation}
The isomorphism with $\SO(3)$ follows from the fact that conjugate subgroups are isomorphic, and we see that each $\boldsymbol{k}$ corresponds to a different copy of $\SO(3)$ embedded as a subgroup the \Poincare group. Under the action $\Sigma$, every fiber $E_m(\boldsymbol{k})$ is a vector space UIR of $H_k \cong \SO(3)$. Vector space UIRs of $\SO(3)$ are uniquely determined up to isomorphism by their dimension $2s+1$ where $s$ is the non-negative integer spin, so the rank of $r$ of the vector bundle $E_m$ determines $s$ via $r=2s+1$.\footnote{$s$ could also be a half integer if we were considering fermions and projective representations.} These vector space representations on the fibers fit together to form a unitary action $\Sigma^{m,S}$ of $H = \SO(3)$ on $E_m$ given by
\begin{equation}
    \Sigma^{m,S}_R(\boldsymbol{k},v) = \Sigma(\Lambda_{\boldsymbol{k}} R \Lambda_{-\boldsymbol{k}})(\boldsymbol{k},v)
\end{equation}
where $R \in \SO(3)$ and $(\boldsymbol{k},v) \in E_m$. This action describes boosting into the particle's rest frame, applying the rotation $R$, and then boosting back to the original momentum. The generator $\boldsymbol{S^m}$ of this action is known in the literature as the Wigner spin operator \cite{Terno2003}. From the vector space representation theory of $\SO(3)$, we can choose a basis $(v_1,\ldots,v_{2s+1})$ for the fiber at $\boldsymbol{k} = 0$ such that the action of $H_{\boldsymbol{0}}$ is generated by the spin $s$ matrices $\boldsymbol{S}_s$. Then we can express $\boldsymbol{S}^m$ in terms of the spin $s$ matrices by 
\begin{equation}\label{eq:spin_s_matrices}
    \boldsymbol{S}^m = \Sigma(\Lambda_{\boldsymbol{k}}) \boldsymbol{S}_s \Sigma(\Lambda_{-\boldsymbol{k}}).
\end{equation}
The form of $\boldsymbol{S}^m$ can be simplified even further by labeling all polarization states by their rest frame polarizations \cite{Wigner1939,Weinberg1995}, that is, we define a basis of $E(\boldsymbol{k})$ at arbitrary $\boldsymbol{k}$ by
\begin{equation}\label{eq:Wigner_basis}
    (\boldsymbol{k},v_i) \doteq \Sigma(\Lambda_{k})(\boldsymbol{0},v_i).
\end{equation}
It follows from Eq. (\ref{eq:spin_s_matrices}) that
\begin{equation}
    \boldsymbol{S}^m(\boldsymbol{k},v_i) = \boldsymbol{S}_s(\boldsymbol{k},v_i).
\end{equation}
so that in this basis we have
\begin{equation}\label{eq:spin_s_matrices_2}
    \boldsymbol{S^m} = \boldsymbol{S}_s.
\end{equation}
We see, either from Eq. (\ref{eq:spin_s_matrices}) or (\ref{eq:spin_s_matrices_2}) or from the fact that $\boldsymbol{S}^m$ generates an $\SO(3)$ symmetry, that  $\boldsymbol{S}^m$ satisfies the $\so(3)$ commutation relations in Eq. (\ref{eq:S_vect_op}). Furthermore, $\boldsymbol{S}^m$ generates an internal symmetry since $\Sigma^S$ does not change $\boldsymbol{k}$. It also is a well-defined vector operator, satisfying Eq. (\ref{eq:vector_operator_condtion}) (\cite{Terno2003}, Eq. (12)). Thus, $\boldsymbol{S}^m$ is a well-defined spin angular momentum operator. Since $\boldsymbol{J}$ and $\boldsymbol{S}^m$ satisfy $\so(3)$ commutation relations, so does 
\begin{equation}\label{eq:Wigner_OAM}
    \boldsymbol{L}^m \doteq \boldsymbol{J} - \boldsymbol{S}^m,
\end{equation}
completing the Wigner SAM-OAM decomposition for massive particles. Note that $\boldsymbol{S}^m$ and $\boldsymbol{L}^m$ also commute with the Hamiltonian \cite{Terno2003}, which follows from the fact that $\boldsymbol{S}^m$ and $\boldsymbol{J}$ (and therefore also $\boldsymbol{L}^m$) generate symmetries which preserve the energy $\sqrt{m^2 + \boldsymbol{k}^2}$. Thus, $\boldsymbol{S}$ and $\boldsymbol{L}$ are conserved in the free theory and, for example, one can use them to label quantum states. We note that it is sometimes stated in the literature that there are issues with relativistic spin and orbital angular momentum even for massive particles \cite{LandauLifshitzQED,PeskinAndSchroeder}, namely, that the spin and orbital angular momentum do not commute with the Hamiltonian and thus do not constitute good quantum numbers. This is true of the Dirac SAM and OAM operators but not of the Wigner SAM and OAM operators; see Ref. \cite{Terno2003} for details of this distinction. The essential idea is that the Dirac spin operator results from working in the bispinor representation in which particle (positive energy) and anti-particle (negative energy) states are treated simultaneously. The total angular momentum operator independently preserves particle and anti-particle states, but the Dirac spin operator mixes the positive and negative energy states and thus does not commute with the Hamiltonian. This does not occur for the Wigner operators $\boldsymbol{S}^m$ and $\boldsymbol{L}^m$ since the bundle $E_m$ contains only particle or anti-particle states (not both simultaneously).

The above procedure breaks down for bundle representations $\pi_0:E_0\rightarrow \Lightcone$ of massless particles. In this case there is no rest frame with $\boldsymbol{k} = 0$, and the little group is now $\ISO(2) = \Real^2 \rtimes \SO(2)$ \cite{Wigner1939,PalmerducaQin_PT}, the group of inhomogeneous isometries of $\Real^2$ where $\rtimes$ denotes the semidirect product. It is well-known that if inhomogeneous elements $\Real^2 \seq \ISO(2)$ act nontrivially, then one obtains particles with an infinite number of internal DOFs \cite{Wigner1939, Weinberg1995}. There are no known particles with such an internal space and these solutions are generally believed to be unphysical. We thus make the standard assumption that the inhomogeneous elements act trivially. We therefore obtain a canonical stabilizing little group representation of $\SO(2)$ on $E_0$. Each fiber is a finite-dimensional representation of $\SO(2)$ and such representations are all one-dimensional and are uniquely labeled by a (possibly negative) integer helicity $h$ rather than a spin \cite{Wigner1939,Weinberg1995,PalmerducaQin_PT}. Since this is not an $\SO(3)$ representation, the generator of this action is not an SAM operator. A clear illustration of this fact is that $\SO(2)$ is a one-dimensional Lie group, so its generator (the helicity operator) is a scalar operator rather than a vector operator, so it obviously cannot represent SAM.

Although this argument is suggestive, it does not rule out the possibility that SAM and OAM operators could be defined by some other procedure. Indeed, the many other proposed SAM and OAM operators are such attempts \cite{Akhiezer1965,Jaffe1990,VanEnk1994,Chen2008,Wakamatsu2010,Bialynicki-Birula2011,Leader2013,Leader2014,Leader2016,Leader2019,Yang2022}. However, we prove in the following theorem that in general there is no SAM-OAM decomposition for massless particles. The argument uses the fundamental difference of the dimensionalities of the spin and helicity representations. Indeed, the vector space spin representations of $\SO(3)$ have dimension $2s+1$ while all of the helicity representations of $\SO(2)$ are one-dimensional. Thus the fibers of massless particles are all one-dimensional, that is, massless particles are represented by line bundles \cite{PalmerducaQin_PT, PalmerducaQin_helicity}. For example, we previously showed that the right (R) and left (L) circularly polarized photons each form a UIR of $\poincare$ on line bundles \cite{PalmerducaQin_PT}, and likewise for the R and L gravitons \cite{PalmerducaQin_GT}.

\begin{noGoTheorem}\label{thm:nogo1}
    Suppose $\pi_0:E_0\rightarrow \Lightcone$ is a massless particle, that is, it is a UIR of $\poincare$. It is not possible to split the total angular momentum $\boldsymbol{J}$ into nonzero SAM and OAM operators $\boldsymbol{S}$ and $\boldsymbol{L}$, where $\boldsymbol{S}$ and $\boldsymbol{L}$ generate $\SO(3)$ symmetries and $\boldsymbol{S}$ corresponds to an internal symmetry.
\end{noGoTheorem}
\begin{proof}
    Suppose such a decomposition exists. $\boldsymbol{S}$ generates an $\SO(3)$ action $\Sigma^S$ on $E_0$. Since this symmetry is internal, $\Sigma^S$ does not change the fiber coordinate $k \in \Lightcone$, and thus $\Sigma^S$ is a vector space representation of $\SO(3)$ on each fiber $E_0(k)$. By the little group construction, we saw that since $E_0$ is a UIR of $\poincare$ it must be a line bundle, so $E_0(k)$ is a one-dimensional vector space. The only representation of $\SO(3)$ on a one-dimensional vector space is the spin $0$ representation, which is trivial in the sense that $\SO(3)$ always acts by the identity. Thus $\Sigma^S = I$ and $\boldsymbol{S} = 0$, and therefore $\boldsymbol{J} = \boldsymbol{L}$. There is thus no nontrivial SAM-OAM decomposition.
\end{proof}

Theorem \ref{thm:nogo1} explains the issues that plague the various proposed SAM and OAM operators for photons. Van Enk and Neinhaus \cite{VanEnk1994} first showed that one attempt to define SAM and OAM for light fail to satisfy $\so(3)$ relations and are therefore not actually SAM and OAM operators; we showed \cite{PalmerducaQin_PT} that the same is true of the decomposition proposed by Bialynicki-Birula and Bialynicka-Birula \cite{Bialynicki-Birula2011}. Recently, Yang et al. \cite{Yang2022} showed through direct calculation that the SAM-OAM decompositions of Jaffe-Manohar \cite{Jaffe1990}, Chen et al. \cite{Chen2008}, and Wakamatsu \cite{Wakamatsu2010} do not satisfy $\so(3)$ relations, leading again to the conclusion that none of these are legitimate SAM and OAM operators. Yang et al. propose two new sets of SAM and OAM operators $(\boldsymbol{S}_M,\boldsymbol{L}_M)$ and $(\boldsymbol{S}_M^{\text{obs}},\boldsymbol{L}_M^{\text{obs}})$; see Table II in Ref. \cite{Yang2022}. In accord with No-Go Theorem \ref{thm:nogo1}, neither of these sets of operators are actually well-defined SAM and OAM operators for photons by the standard definition. The first set satisfy $\so(3)$ relations but, as Yang et al. point out, they are not gauge invariant and are thus not measurable. As such, $\boldsymbol{S}_M$ and $\boldsymbol{L}_M$ are not well-defined vector operators on the space of physical states; this is discussed in more detail in Sec. \ref{sec:gauge_redundancy}. The other set of operators, $(\boldsymbol{S}_M^{\text{obs}},\boldsymbol{L}_M^{\text{obs}})$, are gauge invariant, but the $S_{M,i}^{\text{obs}}$ commute with each other rather than satisfy cyclic $\so(3)$ relations, and are therefore not angular momentum operators by the standard definition.

\section{Nontrivial topology as an obstruction to an SAM-OAM decomposition}\label{sec:no_go_2}
The failure of the angular momentum of massless particles to split into SAM and OAM can be understood both in terms of the geometry and topology of such particles. The obstruction originates in the singular limit that occurs as the particle mass is taken to zero. We showed in the previous section that there is a singularity in the \Poincare geometry that occurs in this limit as the little group jumps abruptly from $\SO(3)$ to $\ISO(2)$. This leads to massive particles possessing spin while massless particles possess helicity, and the latter is not associated with an angular momentum operator. 

However, the $m\rightarrow 0$ limit is also accompanied by a topological singularity as the momentum space $M$ jumps from being a topologically trivial (contractible) mass hyperboloid $\mathcal{M}_m$ to the non-contractible lightcone $\Lightcone$. Since $\mathcal{M}_m$ is contractible to a point, there are only topologically trivial vector bundles over $\mathcal{M}_m$ (\cite{Bott2013}, Corollary 6.9), and thus all massive particles $\pi_m : E_m \rightarrow \mathcal{M}_m$ are described by trivial bundles. However, $\Lightcone$ is not contractible due to the hole at $\boldsymbol{k}=0$, and thus there are topologically nontrivial vector bundles over $\Lightcone$. Indeed, a massless particle with nonzero helicity is described by a vector bundle $\pi_h:\gamma_h \rightarrow \Lightcone$ whose topology is fully characterized by its (first) Chern number $C(\gamma_h) = -2h$ \cite{PalmerducaQin_helicity}. When $h\neq 0$, $\gamma_h$ is topologically nontrivial. This nontrivial topology indicates that the internal DOFs parametrizing the fibers (polarizations) and the external DOFs parametrizing $\Lightcone$ are nontrivially twisted together \cite{Bott1965}. However, SAM and OAM are related to $\SO(3)$ symmetries of the internal and external DOFs of a particle, respectively. Thus, the nontrivial topology of massless particles with nonzero helicity can be heuristically understood as obstructing the existence of such a splitting. As for massive particles with zero spin, the $h=0$ case is uninteresting from the perspective of SAM and OAM as the unique SAM and OAM operators are simply $\boldsymbol{S} = 0$ and $\boldsymbol{L} = \boldsymbol{J}$. 

We now formalize this heuristic observation. In Theorem \ref{thm:nogo1}, we illustrated the geometric obstruction to an SAM-OAM splitting by imposing the condition that the SAM operator is internal. To illustrate the topological obstruction we will instead assume that the OAM operator is external, that is, it generates an action only changing the base manifold DOFs. In precise terms, suppose $\Sigma^L$ is a $G$-action on a bundle $\pi:E\rightarrow M$ with base manifold action $\sigma^L$. We say that $\Sigma^L$ is external if the condition $\sigma^L_{g_1}k = \sigma^L_{g_2}k$ for $g_1,g_2 \in G$ and $k \in M$ implies that $\Sigma^L_{g_1}(k,v) = \Sigma^L_{g_2}(k,v)$ for every $(k,v) \in E(k)$. This is equivalent to saying that if $\sigma^L_g k = k$, then $\Sigma^L_g (k,v) = (k,v)$, i.e., the little group with respect to $\Sigma^L$ acts trivially on $E$. We refer to a Lie algebra representation as external if it generates an external Lie group representation.
\begin{noGoTheorem}\label{thm:nogo2}
    Suppose $\pi:E\rightarrow \Lightcone$ is a rank $r$ bundle with an external $\SO(3)$ action $\Sigma^L$. Assume the corresponding action $\sigma^L$ on $\Lightcone \cong \pspace$ is the standard base manifold action $\sigma^L_R = R$ for $R \in \SO(3)$. Then $E$ is a trivial bundle. Consequently, a massless particle admits an external OAM operator $\boldsymbol{L}$ if and only if the particle has zero helicity.
\end{noGoTheorem}
\begin{proof}
    Consider the restriction of $E$ to the unit sphere $S^2 \seq \pspace \cong \Lightcone$, denoted by $E|_{S^2}$. $\Lightcone$ deformation retracts onto $S^2$, so $E$ is trivial if and only if $E|_{S^2}$ is trivial (\cite{HatcherVBKT}, Cor 1.8). By the assumption on $\sigma^L$, $\Sigma^L$ restricts to an action on $E|_{S^2}$. Fix some $k_0 \in S^2$ and for each $k \in S^2$, define
    \begin{equation}
        G_k = \{R \in \SO(3)|Rk_0 = k\}.
    \end{equation}
    Let $v_1,...,v_r$ be a basis of $E|_S^2$ at $k_0$, and define sections $f_1,...,f_r$ of $E|_{S^2}$ by 
    \begin{equation}
        f_j(k) \doteq \Sigma^L_{G_k}(k_0,v_j).
    \end{equation}
    The assumption that $\Sigma^L$ is external ensures that all elements of $G_k$ have the same action on $(k_0,v_j)$, so the $f_j$ are well-defined. Since $\Sigma^L$ is smooth and $G_k$ varies smoothly with $k$, the $f_j$ are smooth sections of $E|_{S^2}$. Furthermore, they are linearly independent at each $k$. To see this, choose some $g_k \in G_k$ and note that $\Sigma^L_{g_k}$ is a linear isomorphism between the fibers at $k_0$ and at $k$. Since $E$ possess $r$ linearly independent sections, it is a trivial vector bundle. Since a massless particle with nonzero helicity is topologically nontrivial \cite{PalmerducaQin_helicity}, such a particle does not admit an external OAM operator.
\end{proof}

No-Go Theorem \ref{thm:nogo2} shows that a massless helicity $h\neq 0$ particle $\pi_h:\gamma_h \rightarrow \mathcal{L}_+$ does not admit an OAM operator since $\gamma_h$ is a nontrivial bundle. 

In practice, it is common to consider helicity $h$ particles not in isolation, but in combination with the helicity $-h$ particles, that is, one considers the direct sum bundle $\gamma_{h,T} \doteq \gamma_h \oplus \gamma_{-h}$. For example, R and L photons are respectively described by bundles $\gamma_1$ and $\gamma_{-1}$ \cite{PalmerducaQin_PT}, but it is necessary to consider the total photon bundle $\gamma_T \doteq \gamma_{1} \oplus \gamma_{-1}$ if one wishes to describe linear polarizations. Since the first Chern numbers are additive with respect to the direct sum by the Whitney product formula \cite{Tu2017differential} and since $C(\gamma_h) = -2h$ \cite{PalmerducaQin_helicity}, it follows that $C(\gamma_{h,T}) = 0$. This implies that, unlike $\gamma_{\pm h}$, the bundle $\gamma_{h,T}$ is topologically trivial \cite{PalmerducaQin_PT, PalmerducaQin_helicity}. No-Go Theorem \ref{thm:nogo2} does not then immediately imply that $\gamma_{h,T}$ does not possess an OAM operator, although this is indeed the case.

Before proving this, we will give an informal argument showing why one would not expect an OAM operator to exist for $\gamma_{h,T}$. If we let $\Sigma^\pm$ denote the $\poincare$ symmetries of $\gamma_{\pm h}$, then $\Sigma = \Sigma^+ \oplus \Sigma^-$ is an $\poincare$ symmetry of $\gamma_{h,T} = \gamma_h \oplus \gamma_{-h}$. Suppose there were such an OAM operator $\boldsymbol{L}^T$ which generates an $\SO(3)$ symmetry $\Sigma^{L,T}$ on $\gamma_{h,T}$. By No-Go Theorem $\ref{thm:nogo2}$, $\Sigma^{L,T}$ is not a symmetry of either $\gamma_h$ or $\gamma_{-h}$, but only of their direct sum. However, $\gamma_{\pm h}$ are preserved under the actions $\Sigma^\pm$, so $\boldsymbol{L}^T$ cannot be written purely in terms of the generators of $\Sigma^\pm$, that is, in terms of the Hamiltonian, momentum, angular momentum, or boost operators on $\gamma_{\pm h}$. There are no other natural operators available, so this gives a strong indication that there is no way to construct an OAM operator $\boldsymbol{L}^T$. Indeed, $\gamma_{h}$ and $\gamma_{-h}$ are related by a parity symmetry $P$ which is part of the full \Poincare group, but not of the proper orthochronous \Poincare group $\poincare$ \cite{PalmerducaQin_PT}. Since $\boldsymbol{L}^T$ would have to generate a symmetry mixing $\gamma_h$ and $\gamma_{-h}$, this suggests $\boldsymbol{L}^T$ would involve the generator of the parity symmetry $P$. However, parity is a discrete rather than a continuous symmetry, and thus has no infinitesimal generator. This is essentially the reason why elementary particles are typically defined \cite{Weinberg1995} as UIRs of $\poincare$ rather than the full \Poincare group, since the Lie algebra does not detect discrete time and parity symmetries.

We can give a more rigorous argument showing that no such $\boldsymbol{L}^T$ exists as follows. Note that while No-Go Theorems \ref{thm:nogo1} and $\ref{thm:nogo2}$ are proved via  different methods, their conclusions are related. In particular, if the OAM operator $\boldsymbol{L}^T$ is external, then $\boldsymbol{S}^T = \boldsymbol{J}^T-\boldsymbol{L}^T$ is internal. We note then that No-Go Theorem $\ref{thm:nogo1}$ applies equally well to the bundle $\gamma_{h,T}=\gamma_h \oplus \gamma_{-h}$. The only change in the proof of No-Go Theorem \ref{thm:nogo1} is that the fibers now have dimension two rather than one. However, there are no two-dimensional UIRs of $\SO(3)$ and thus the representation generated by an internal $\boldsymbol{S}^T$ must break down into two one-dimensional UIRs of $\SO(3)$ \footnote{There is a two-dimensional projective UIR corresponding to spin 1/2 fermions, but this is irrelevant to the case of bosons, whose representations are non-projective.}. But the only such UIR is the trivial representation, so we must have $\boldsymbol{S}^T = 0$ and $\boldsymbol{L}^T =\boldsymbol{J}^T = \boldsymbol{J}^{(+h)} + \boldsymbol{J}^{(-h)}$, where $\boldsymbol{J}^{(\pm h)}$ are the total angular momentum operators of $\gamma_{\pm h}$. However, $\boldsymbol{J}^{(\pm h)}$, and therefore $\boldsymbol{J}^T$, do not generate external symmetries when $h \neq 0$ since the little group acts nontrivially, and thus, $\gamma_{h,T}$ does not admit an external OAM operator. We have thus proved the following corollary of the previous no-go theorems:
\begin{corollary}
    Let $\gamma_{\pm h}$ be line bundles describing massless particles with nonzero helicities $\pm h$ and with total angular momentum operators $\boldsymbol{J}^\pm$. Consider the total bundle $\gamma_{T,h} = \gamma_{h} \oplus \gamma_{-h}$
    with the induced angular momentum operator $\boldsymbol{J}^T = \boldsymbol{J}^{(+h)} + \boldsymbol{J}^{(-h)}$. If
    \begin{equation}
        \boldsymbol{J}^T = \boldsymbol{L}^T + \boldsymbol{S}^T,
    \end{equation}
    where $\boldsymbol{L}^T$ and $\boldsymbol{S}^T$ generate $\SO(3)$ symmetries, then:
    \begin{enumerate}
        \item If $\boldsymbol{S}^T$ is internal, then $\boldsymbol{S}^T = 0$.
        \item $\boldsymbol{L}^T$ cannot generate an external symmetry.
    \end{enumerate}
    Therefore, there is no nontrivial SAM-OAM decomposition for the system of particles $\gamma_T$.
\end{corollary}

\section{The angular momentum cannot be split into non-SO(3) symmetries}\label{sec:no_go_3}
No-Go Theorems \ref{thm:nogo1} and \ref{thm:nogo2} establish the fundamental result that it is not possible to split the angular momentum of a massless particle into genuine SAM and OAM operators. This reflects the fact that the internal and external DOFs cannot independently carry any nontrivial $\SO(3)$ symmetries. One might still ask if $\boldsymbol{J}$ can be split into vector operators describing some non-$\SO(3)$ symmetries.

As an illustration of why this is not possible, consider the proposed splitting \cite{VanEnk1994, Bialynicki-Birula2011}
\begin{align}
    \boldsymbol{J} &= \boldsymbol{J}_\parallel + \boldsymbol{J}_\perp \label{eq:helicity_decomp}\\
    \boldsymbol{J}_\parallel &\doteq (\boldsymbol{J}\cdot \hat{\bfk}) \hat{\bfk} \label{eq:spin_attempt}\\
    \boldsymbol{J}_\perp &\doteq \boldsymbol{J} - \boldsymbol{J}_\parallel. \label{eq:orbital_attempt}
\end{align}
$\boldsymbol{J} \cdot \hat{\bfk}$ is the helicity operator generating the internal $\SO(2)$-action which naturally arises from the little group construction for massless particles \cite{PalmerducaQin_PT}. The helicity is a scalar operator rather than a vector operator, and thus by itself cannot be a SAM operator. Equation (\ref{eq:spin_attempt}) can be viewed as an ad hoc attempt to ``vectorize'' the helicity operator into an SAM operator. However, by Theorem \ref{thm:nogo1} this $\boldsymbol{J}_\parallel$ cannot actually be an angular momentum operator. Indeed, while $\boldsymbol{J}_\parallel$ and $\boldsymbol{J}_\perp$ are well-defined vector operators satisfying condition (\ref{eq:vector_operator_condtion}), they satisfy the nonstandard commutation relations \cite{VanEnk1994,PalmerducaQin_PT}
\begin{align}
    [J_{\parallel,m},J_{\parallel,n}]&=0 \label{eq:Lie_1},\\ 
    [J_{\perp,m},J_{\parallel,n}]&= i \epsilon_{mnp}J_{\parallel,p} \label{eq:Lie_2},\\
    [J_{\perp,m},J_{\perp,n}]&= i \epsilon_{mnp}(J_{\perp,p}-J_{\parallel,p}) \label{eq:Lie_3}.
\end{align}
Aside from not generating $\SO(3)$ symmetries, we previously noted that these suffer from an additional issue \cite{PalmerducaQin_PT}. In particular, while the $J_{\parallel,m}$ operators commute with each other and thus generate an $\Real^3$ symmetry, the $J_{\parallel,p}$ on the rhs of Equation $(\ref{eq:Lie_3})$ shows that the components of $\boldsymbol{J}_\perp$ do not combine together to form a Lie subalgebra. Thus, $\boldsymbol{J}_\perp$ does not generate any symmetry at all. The next theorem shows that this issue is generic---it arises in any attempt to split $\boldsymbol{J}$ into generators of non-$\SO(3)$ symmetries, regardless of whether they generate internal and/or external symmetries.

\begin{noGoTheorem}\label{thm:nogo3}
    Suppose that
    \begin{equation}
        \boldsymbol{J} = \boldsymbol{L} + \boldsymbol{S}
    \end{equation}
    where $\boldsymbol{J}$ is the the total angular momentum operator and $\boldsymbol{L}$ and $\boldsymbol{S}$ are nonzero vector operators, that is, they satisfy
    \begin{subequations}\label{eq:ng2_vector_operator}
    \begin{align}
        [L_m,J_n] &= i\epsilon_{mnp}L_p \\
        [S_m,J_n] &= i\epsilon_{mnp}S_p.
    \end{align}
    \end{subequations}
    Furthermore, suppose that $\boldsymbol{L}$ and $\boldsymbol{S}$ are generators of symmetry groups $G_L$ and $G_S$ with associated Lie algebras $\mathfrak{g}_L$ and $\mathfrak{g}_S$. Then
    \begin{equation}
        \mathfrak{g}_L = \mathfrak{g}_S = \so(3).
    \end{equation}
    In particular, it is not possible for either $\boldsymbol{L}$ or $\boldsymbol{S}$ to generate non-$\SO(3)$ symmetries.
\end{noGoTheorem}
\begin{proof}
     $\boldsymbol{L}$ and $\boldsymbol{S}$ satisfy the commutation relations of their associated Lie algebras $\mathfrak{g}_L$ and $\mathfrak{g}_S$. Therefore,
    \begin{align}
        [L_m,L_n] = i \nu_{mnp}L_p \label{eq:ng2_L_comm}\\
        [S_m,S_n] = i \tau_{mnp}S_p \label{eq:ng2_S_comm}
    \end{align}
    where $\nu_{mnp}$ and $\tau_{mnp}$ are the structure constants of $\mathfrak{g}_L$ and $\mathfrak{g}_S$. It then follows that
    \begin{align}
        [L_m,L_n] &= [J_m - S_m, J_n - S_n] \\
        &= i\epsilon_{mnp}L_p + i(\tau_{mnp}-\epsilon_{mnp})S_p \\
        &= i\nu_{mnp}L_p.
    \end{align}
    Thus, either $\tau_{mnp} = \epsilon_{mnp}$ or else $S_p$ is a multiple of $L_p$. If $\tau_{mnp} = \epsilon_{mnp}$, then also $\nu_{mnp} = \epsilon_{mnp}$, and therefore $\boldsymbol{L}$ and $\boldsymbol{S}$ both satisfy $\so(3)$ commutation relations, proving the claim. Suppose then that $S_m = c_m L_m$ for some constants $c_m$. Note that $c_m \neq -1$ since $J_m \neq 0$. Then $L_m = (1+c_m)^{-1}J_m$. From Equation (\ref{eq:ng2_vector_operator}), it follows that
    \begin{align}
        [L_m,J_n] &= \frac{1+c_p}{1+c_m}i\epsilon_{mnp}L_p \\ 
        &= i\epsilon_{mnp}L_p,
    \end{align}
    so $c_m = c_p \doteq c$ for all $m,p$. Then
    \begin{align}
        [L_m,L_n] &= \frac{i}{1+c}\epsilon_{mnp}L_p.
    \end{align}
    By rescaling the basis to $L_m' = (1+c)L_m$ we obtain
    \begin{equation}
        [L_m',L_n'] = i\epsilon_{mnp}L_p',
    \end{equation}
    showing that $\mathfrak{g}_L = \so(3)$. An analogous argument shows that $\mathfrak{g}_S = \so(3)$.
\end{proof}
No-Go Theorem \ref{thm:nogo3} shows that if $\boldsymbol{S}$ or $\boldsymbol{L}$ generates a non-$\SO(3)$ symmetry, then the other operator will not generate any symmetry at all. Thus, such splittings of the angular momentum appear ill-suited to generalize SAM and OAM operators. 

We note that while the vector operator $\boldsymbol{J}_\perp$ in the splitting (\ref{eq:helicity_decomp})-(\ref{eq:orbital_attempt}) does not generate any 3D symmetry, any individual component, say $J_{\perp,3}$, commutes with itself. It therefore generates a 1D symmetry, and such symmetries can be of interest in interacting systems if they commute with the Hamiltonian, as discussed in Ref. \cite{FernandezCorbaton2014}.

\section{Generalization to non-internal SAM operators}\label{sec:no_go_4}
In this section we prove a stronger version of No-Go Theorem \ref{thm:nogo1} by relaxing the constraint that the SAM operator is stabilizing. Assume again that $\pi_0:E_0 \rightarrow \Lightcone$ and $\Sigma$ are a representation of a massless particle, and that
\begin{equation}\label{eq:ng3_splitting}
    \boldsymbol{J} = \boldsymbol{L} + \boldsymbol{S}
\end{equation}
where $\boldsymbol{L}$ and $\boldsymbol{S}$ are vector operators generating symmetries. By No-Go Theorem \ref{thm:nogo3}, these must generate $\SO(3)$ symmetries $\Sigma^{L}$ and $\Sigma^S$ on $E$. Let $\sigma$, $\sigma^L$, and $\sigma^S$ be the actions on the base manifold $\Lightcone$ corresponding to $\Sigma$, $\Sigma^L$, and $\Sigma^S$. For massless particles, the energy is given by $\omega = |\boldsymbol{k}|$. If one enforces the constraint that $\Sigma^S$ is internal, then $\sigma^S_R = I$ and $\sigma^L_R = \sigma_R = R$ for any $R \in \SO(3)$. In this case,
\begin{equation}\label{eq:energy_perserving}
|\sigma^{L}_R\boldsymbol{k}| = |\sigma^{S}_R\boldsymbol{k}| = |\boldsymbol{k}|,
\end{equation} 
so $\boldsymbol{L}$ and $\boldsymbol{S}$ both generate symmetries which preserve the energy. The first no-go theorem shows that there are no nontrivial splittings such that $\boldsymbol{S}$ is internal. We now prove this conclusion is still true if we replace the assumption that $\boldsymbol{S}$ is internal with the weaker assumption that $\boldsymbol{S}$ and $\boldsymbol{L}$ conserve energy in the sense of (\ref{eq:energy_perserving}). To this end, we first establish the following lemmas.
\begin{lemma}\label{lm:commutative}
    Suppose $\boldsymbol{L}$ and $\boldsymbol{S}$ are vector operators satisfying equation (\ref{eq:ng3_splitting}). Then $\boldsymbol{L}$ and $\boldsymbol{S}$ commute, and therefore so do $\sigma^L$ and $\sigma^S$.
\end{lemma}
\begin{proof}
    From equations (\ref{eq:ng3_splitting}) and (\ref{eq:L_vect_op}) we have that
    \begin{align}
        [L_m,J_n] &= i\epsilon_{mnp} L_p + [L_m, S_n] \\
                  &= i\epsilon_{mnp} L_p.
    \end{align}
    Thus 
    \begin{equation}
        [L_m, S_n] = 0,
    \end{equation}
    and since the generators commute, so do the generated actions $\Sigma^L$ and $\Sigma^S$. This in turn implies that $\sigma^L$ and $\sigma^S$ commute. 
\end{proof}

Let $\diff(M)$ be the group of diffeomorphisms on a manifold $M$. A $G$-action 
\begin{gather}
    \zeta:G\rightarrow \diff(M) \\
    g \mapsto \zeta_g
\end{gather}
on a manifold $M$ is a continuous group homomorphism from the Lie group $G$ to $\diff(M)$. $\zeta$ is transitive if for any $m_1, m_2 \in M$, there is a $g \in G$ such that $\zeta_{g}m_1 = m_2$. Under the energy conserving assumption (\ref{eq:energy_perserving}), $\sigma$, $\sigma^L$, and $\sigma^S$ are all $\SO(3)$-actions on the unit sphere $S^2 \seq \Lightcone$. 

\begin{lemma}\label{lm:transative_sphere_action}
    Suppose that $\zeta:\SO(3) \rightarrow \diff(S^2)$ is an $\SO(3)$-action on $S^2$. Then
    \begin{enumerate}[(i)]
        \item Either $\zeta$ is the identity operator or else it is transitive.
        \item For every $R \in \SO(3)$, $\zeta_R$ has a fixed point, that is, there exists an $x \in S^2$ such that $\zeta_R x = x$.
    \end{enumerate}
\end{lemma}
\begin{proof}
    (i) \cite{Khan_MSE} For $x \in S^2$, let $G_x$ be the isotropy group of $x$ with respect to $\zeta$. $G_x$ is a closed subgroup of $\SO(3)$. The only closed subgroups of $\SO(3)$ are isomorphic to either a finite group, $\SO(2)$, $\mathrm{O}(2)$, or $\SO(3)$ \cite{Antoneli2012}. 
    Consider the map $F:\SO(3)/G_x \rightarrow S^2$ given by $F(R) = \zeta_R x$. This gives a homeomorphism of $\SO(3)/G_x$ with the orbit of $x$ in $S^2$ \cite{Frankel2011}. If $G_x$ is a finite group, then $\SO(3)/G_x$ would be three-dimensional but homeomorphic to a subset of $S^2$, which is an obvious contradiction. Similarly, $\SO(3)/\mathrm{O}(2) \cong \mathbb{RP}^2$, but the latter cannot be embedded in $S^2$ \cite{Tu_manifolds}, so $G_x \not \cong \mathrm{O}(2)$. If $G_x \cong \SO(2)$, the orbit of $x$ is homeomorphic to $\SO(3)/\SO(2) \cong S^2$ \cite{Frankel2011}. There are no proper subspaces of $S^2$ homeomorphic to $S^2$, thus the orbit of $x$ must be all of $S^2$. Thus, if $G_x \cong \SO(2)$ for some $x$, then $\zeta$ is transitive. 
    
    Lastly, if $G_x = \SO(3)$, then $\zeta_R x = x$ for all $R \in \SO(3)$. Furthermore, this implies that $\zeta_R y = y$ for all $R \in \SO(3)$ and $y \in S^2$. For suppose there were a $y$ for which this were not true. Then $G_y \cong \SO(2)$ and the action is transitive and thus $x$ would be in the orbit of $y$, which is a contradiction. Thus, if there exists any $x$ for which $G_x = \SO(3)$, then $\zeta$ is just the identity action on $S^2$.

    (ii) Let $I$ be the identity in $\SO(3)$, and note that $\zeta_I =\mathds{1}$ is the identity operator on $S^2$. Since $\SO(3)$ is path-connected, $\zeta_R:S^2 \rightarrow S^2$ is homotopic to the identity for any $R \in \SO(3)$. The degree of a map $f:S^2 \rightarrow S^2$ is homotopy invariant \cite{HatcherAT}, so $\text{deg}(\zeta_R) = \text{deg}(\mathds{1}) = 1$. If a map $f:S^2\rightarrow S^2$ has no fixed points, then $\text{deg}(f) = -1$ \cite{HatcherAT}. Thus $\zeta_R$ must have a fixed point.
\end{proof}

\begin{noGoTheorem}\label{thm:nogo4}
    Assume that the $\SO(3)$ action $\Sigma$ of a massless particle splits into UIRs of $\SO(3)$, $\Sigma^L$ and $\Sigma^S$, which are generated by the vector operators $\boldsymbol{L}$ and $\boldsymbol{S}$ as in Eq. (\ref{eq:ng3_splitting}). Assume further that these actions preserve energy in the sense of Eq. (\ref{eq:energy_perserving}). Then either $\Sigma^L$ or $\Sigma^S$ acts trivially. Thus, either $\boldsymbol{L}=0$ or $\boldsymbol{S}=0$, and consequently the splitting is trivial.
\end{noGoTheorem}
\begin{proof}
    If either of $\Sigma^L$ or $\Sigma^S$ is an internal action then the result follows from No-Go Theorem \ref{thm:nogo1}, so consider the case when neither $\Sigma^L$ or $\Sigma^S$ are internal. Let $E_{S^2_r}$ be the bundle obtained by restricting the base manifold of $E_0$ to $S_r^2$ where $S^2_r \seq \pspace \cong \Lightcone$ is the sphere of radius $r>0$. By the energy preserving condition, $\Sigma^L$ and $\Sigma^S$ restrict to $\SO(3)$-actions on $E_{S^2_r}$. Consider the corresponding actions $\sigma^{L,r}$ and $\sigma^{S,r}$ on the base manifold $S^2_r \cong S^2$. By Lemma 2 (i), for each $r$, $\sigma^{L,r}$ is either transitive or the identity. Since $\Sigma^L$ is not internal, there must be some $r_0$ such that $\sigma^{L,r_0}$ is transitive. Furthermore, because the action $\sigma^L$ is continuous, by varying $r_0$ to any other $r>0$, we obtain that $\sigma^{L,r}$ is also transitive. The same argument shows that $\sigma^{S,r}$ is transitive for every $r$.
    
    Fix some $r>0$. By Lemma \ref{lm:commutative},
    \begin{equation}\label{eq:commute}
        \sigma^{S,r}_{R_1}\sigma^{L,r}_{R_2} = \sigma^{L,r}_{R_2}\sigma^{S,r}_{R_1}
    \end{equation}
    for every $R_1,R_2 \in \SO(3)$. Consider an arbitrary but fixed $R_1$. By Lemma \ref{lm:transative_sphere_action} (ii), $\sigma^{S,r}_{R_1}$ has some fixed point $y \in S^2_{r}$. By (\ref{eq:commute}),
    \begin{equation}
        \sigma^{S,r}_{R_1}(\sigma^{L,r}_{R_2}y) = \sigma^{L,r}_{R_2}y,
    \end{equation}
    showing that $\sigma^{L,r}_{R_2}y$ is a fixed point of $\sigma^{S,r}_{R_1}$ for every $R_2$. Since $\sigma^{L,r}$ is transitive, $\sigma^{L,r}_{R_2}y$ spans all of $S^2_{r}$ as $R_2$ is varied over $\SO(3)$. Therefore, every point is a fixed point of $\sigma^{S,r}_{R_1}$, and thus $\sigma^{S,r}_{R_1}$ is the identity. Since this argument can be applied to every $R_1$, it follows that $\sigma^{S,r}$ is the trivial action. However, $\sigma^{S,r}$ was nontrivial by assumption, so a contradiction has been reached, showing that either $\Sigma^L$ or $\Sigma^S$ must be trivial.
\end{proof}

\section{Non-physical SAM operators arising from gauge redundancy}\label{sec:gauge_redundancy}
In this section we give an example of how gauge redundancies in the description of massless particles can lead to the identification of non-physical SAM operators for massless particles. For concreteness, we consider the case of photons. It is common to describe photons in momentum space using the Lorentz gauge vector potential $A^\mu(k) = (A^0(k), \boldsymbol{A}(k))$, which satisfies the Lorentz gauge condition
\begin{equation}\label{eq:Lorentz_gauge}
    A^\mu k_\mu = 0.
\end{equation} 
Note that while the Lorentz gauge has the advantage of being Lorentz covariant, it is not complete, i.e., there is still residual redundancy in this description \cite{Tong2006}. Condition (\ref{eq:Lorentz_gauge}) determines $A^0$ in terms of $k$ and $\boldsymbol{A}$, so the potential can be parameterized just by $\boldsymbol{A}(k)$. One can then form the photon bundle $\pi_\gamma: \gamma \rightarrow \Lightcone$ as the collection of all $(k,\boldsymbol{A})$ where $k \in \Lightcone$ and $\boldsymbol{A} \in \Comp^3$. Notice that this bundle is topologically trivial $\gamma \cong \Lightcone \times \Comp^3$ since the fiber does not depend on $k$. As such, it is similar to a mass $m>0$ particle with spin 1, which is described by the vector bundle $M_m \times \Comp^3$. In particular, one can define an SAM operator $\boldsymbol{S}^\gamma$ for $\gamma$ which acts on the fibers by the spin 1 matrices
\begin{equation} \label{eq:spin_1_matrices}
    {(S^{\gamma}_{n})}_{pq} = -i \epsilon_{npq}.
\end{equation}
By defining $\boldsymbol{L}^\gamma = \boldsymbol{J} - \boldsymbol{S}^\gamma$, which generates rotations in $M_m$ without affecting $\Comp^3$, it might then seem that we have obtained an SAM-OAM decomposition of the photon angular momentum. Indeed, this was one of the early attempts at such a splitting \cite{Akhiezer1965, VanEnk1994}. However, this splitting is non-physical and is only possible due to residual redundancy in the Lorentz gauge. To see this, let $\Sigma^{\gamma,S}$ denote the $\SO(3)$ symmetry on $\gamma$ generated by $\boldsymbol{S}^\gamma$. $\gamma$ is a rank-3 bundle, but there are only two photon polarizations. A typical way of removing the remaining gauge freedom is to impose the Coulomb gauge condition $\boldsymbol{k} \cdot \boldsymbol{A} = 0$. Applying this condition in each fiber gives the rank-2 Coulomb gauge photon bundle in which all vectors represent physically different states \cite{PalmerducaQin_GT}. More generally, the remaining gauge freedom can be removed without choosing a specific gauge via the procedure of Asorey et al. \cite{Asorey1985}. At each fixed $k \in \Lightcone$, $\boldsymbol{A}$ and $\boldsymbol{A} + \alpha \boldsymbol{k}$ are physically equivalent for any $\alpha \in \Comp$, and thus we can use this to define an equivalence relation $\sim$ on each fiber. We then remove all gauge redundancy by forming the rank-2 vector bundle $\gamma/\sim$ by taking the quotient in each fiber. Regardless of how the gauge redundancy is removed, one obtains a rank-2 vector bundle $\tilde{\gamma}$ of physical states. If $\boldsymbol{S}$ is a true SAM operator, it must descend to an operator on $\tilde{\gamma}$, generating an $\SO(3)$ action $\Sigma^{\tilde{\gamma},S}$ on $\tilde{\gamma}$. This is equivalent to saying that if $(k, \boldsymbol{A})$ and $(k,\boldsymbol{A}')$ are physically equivalent, then $\Sigma^{\gamma,S}_R(k,\boldsymbol{A})$ and $\Sigma^{\gamma,S}_R (k,\boldsymbol{A}')$ must be physically equivalent for any $R \in \SO(3)$. However, this is not possible for the same dimensional reasons that lead to No-Go Theorem \ref{thm:nogo1}. In particular, $\Sigma^{\tilde{\gamma},S}$ would be an $\SO(3)$ representation on the two-dimensional fibers of $\tilde{\gamma}$, but the only such representation is the trivial representation. Thus $\boldsymbol{S}^\gamma$ would also have to be trivial, contradicting Equation (\ref{eq:spin_1_matrices}). Therefore $\boldsymbol{S}^\gamma$ cannot descend to a well-defined operator on the space of physical photon states; it is an artifact of gauge redundancy in the Lorentz gauge. Similarly, massless particles characterized by helicity $\pm h$ with $h>0$ can together be described by a spin $h$ representation \cite{Asorey1985}. However, the spin $h$ representation has dimension $2h + 1$, while the the $\pm h$ massless particles only have two independent internal degrees of freedom, and thus there is a dimension $2h-1$ gauge redundancy in this description. While this gauge redundant description might appear to give rise to an SAM-OAM decomposition, this decomposition cannot be gauge independent and is thus not physical.

\section{Single-particle states vs. Fock space}\label{sec:Fock_space}
We note that we have obtained the No-Go Theorems by examining the symmetries of the single-particle states rather than the multiparticle Fock space. Under very mild physical assumptions, these treatments are equivalent for the SAM-OAM decomposition problem. Let $\mathcal{H} = L^2(E)$ be the single-particle Hilbert space. The boson Fock space $\mathcal{F}$ is constructed from symmetrized sums of copies of $\mathcal{H}$ as \cite{Folland2008}
\begin{equation}
  \mathcal{F} = \bigoplus_{n=0}^{\infty}\text{Sym}^n(\mathcal{H}).
\end{equation}
Thus, any symmetry of the single-particle space $\mathcal{H}$ induces a corresponding symmetry of $\mathcal{F}$ given by applying the symmetry to each copy of $\mathcal{H}$. A priori, there is the possibility that there exist $\SO(3)$ symmetries of the Fock space that do not arise in this manner, that is, actions of $\SO(3)$ which mix states of different particle number. Our analysis has not accounted for such symmetries. However, it is reasonable to make the assumption that such exotic symmetries, if they exist, do not describe SAM or OAM operators. Indeed, the total angular momentum operator does not mix states of different particle number, nor do any of the previously proposed SAM or OAM operators for massless particles. For example, see Yang et al.'s \cite{Yang2022} review of the various proposed SAM and OAM operators, and note that in every term creation and annihilation operators appear in equal numbers so that particle number is conserved. Indeed, it seems to be implicitly assumed in discussions of the SAM or OAM of elementary particles that these properties are defined at the single-particle level and are not emergent multiparticle phenomena. We also make this assumption.

\section{Conclusion}
The fundamental results of this paper are No-Go Theorems \ref{thm:nogo1} and \ref{thm:nogo2}, showing that the total angular momentum operator of massless bosons cannot be split into SAM and OAM operators. This result may be surprising given that the angular momentum of massive particles splits so readily into SAM and OAM. On the other hand, this result could have been anticipated. Despite the relevancy of the OAM and SAM of light to optics and photonics and despite sustained interest in this problem, there is no agreement in the literature as to the correct definition of such OAM and SAM operators. This suggests the existence of a fundamental obstruction to such a splitting, one which we showed traces back to the singular limit that occurs as the mass of a particle goes to zero. This limit entails a topological singularity, as the momentum space jumps abruptly from a contractible mass hyperboloid to the non-contractible lightcone and the topologically trivial rank-$(2s+1)$ massive vector bundles are replaced by nontrivial massless line bundles. This topological singularity is accompanied by a geometric singularity as the little group jumps from $\SO(3)$ for massive particles to $\SO(2)$ for massless particles. Particles carry a canonical internal little group symmetry. As SAM is associated with an internal $\SO(3)$ symmetry rather than an $\SO(2)$ symmetry, this singularity can be viewed as the essential obstruction to an SAM-OAM decomposition.

Furthermore, we have shown that this obstruction is not easily circumvented by generalizing the definition of angular momentum operators. Perhaps the simplest generalization is to allow spin to describe non-internal symmetries, however, No-Go Theorem \ref{thm:nogo4} shows that this does not yield nontrivial operators. Another option is to attempt to define SAM and OAM operators which generate non-rotational symmetries; while this is less natural, it has been attempted. Indeed, as Yang et al. \cite{Yang2022} showed, most of the proposed SAM and OAM operators do not actually generate 3D rotations. No-Go Theorem \ref{thm:nogo3} shows that in these cases it is not possible for both the SAM and OAM to generate \emph{any} symmetries, even non-$\SO(3)$ symmetries.

These no-go theorems raise the following question. Although photon angular momentum does not split into SAM and OAM, why is it that OAM beams have found such remarkable success in experimental optics? The answer to this question certainly requires additional theoretical work and examination of specific experimental results. However, we suspect that an important point in addressing this question is that while photon angular momentum cannot be split into SAM and OAM, photons can still carry an arbitrarily large amount of angular momentum. We thus hypothesize that many of the useful features of so-called OAM beams result because they carry large amounts of angular momentum, not necessarily because that angular momentum is orbital.

We also note that as we have considered only non-projective representations, this analysis does not directly apply to Weyl fermions. Nevertheless, we suspect that the no-go theorems can be extended to projective representations, and thus to Weyl fermions; future research will explore the possibility of such extensions.

\begin{acknowledgments}
We thank Moishe Kohan for his help in proving Lemma \ref{lm:transative_sphere_action} (i). This work is supported by U.S. Department of Energy (DE-AC02-09CH11466).
\end{acknowledgments}

\bibliography{SAM_OAM}

\begin{thebibliography}{51}%
\makeatletter
\providecommand \@ifxundefined [1]{%
 \@ifx{#1\undefined}
}%
\providecommand \@ifnum [1]{%
 \ifnum #1\expandafter \@firstoftwo
 \else \expandafter \@secondoftwo
 \fi
}%
\providecommand \@ifx [1]{%
 \ifx #1\expandafter \@firstoftwo
 \else \expandafter \@secondoftwo
 \fi
}%
\providecommand \natexlab [1]{#1}%
\providecommand \enquote  [1]{``#1''}%
\providecommand \bibnamefont  [1]{#1}%
\providecommand \bibfnamefont [1]{#1}%
\providecommand \citenamefont [1]{#1}%
\providecommand \href@noop [0]{\@secondoftwo}%
\providecommand \href [0]{\begingroup \@sanitize@url \@href}%
\providecommand \@href[1]{\@@startlink{#1}\@@href}%
\providecommand \@@href[1]{\endgroup#1\@@endlink}%
\providecommand \@sanitize@url [0]{\catcode `\\12\catcode `\$12\catcode `\&12\catcode `\#12\catcode `\^12\catcode `\_12\catcode `\%12\relax}%
\providecommand \@@startlink[1]{}%
\providecommand \@@endlink[0]{}%
\providecommand \url  [0]{\begingroup\@sanitize@url \@url }%
\providecommand \@url [1]{\endgroup\@href {#1}{\urlprefix }}%
\providecommand \urlprefix  [0]{URL }%
\providecommand \Eprint [0]{\href }%
\providecommand \doibase [0]{https://doi.org/}%
\providecommand \selectlanguage [0]{\@gobble}%
\providecommand \bibinfo  [0]{\@secondoftwo}%
\providecommand \bibfield  [0]{\@secondoftwo}%
\providecommand \translation [1]{[#1]}%
\providecommand \BibitemOpen [0]{}%
\providecommand \bibitemStop [0]{}%
\providecommand \bibitemNoStop [0]{.\EOS\space}%
\providecommand \EOS [0]{\spacefactor3000\relax}%
\providecommand \BibitemShut  [1]{\csname bibitem#1\endcsname}%
\let\auto@bib@innerbib\@empty
\bibitem [{\citenamefont {Shen}\ \emph {et~al.}(2019)\citenamefont {Shen}, \citenamefont {Wang}, \citenamefont {Xie}, \citenamefont {Min}, \citenamefont {Fu}, \citenamefont {Liu}, \citenamefont {Gong},\ and\ \citenamefont {Yuan}}]{Shen2019}%
  \BibitemOpen
  \bibfield  {author} {\bibinfo {author} {\bibfnamefont {Y.}~\bibnamefont {Shen}}, \bibinfo {author} {\bibfnamefont {X.}~\bibnamefont {Wang}}, \bibinfo {author} {\bibfnamefont {Z.}~\bibnamefont {Xie}}, \bibinfo {author} {\bibfnamefont {C.}~\bibnamefont {Min}}, \bibinfo {author} {\bibfnamefont {X.}~\bibnamefont {Fu}}, \bibinfo {author} {\bibfnamefont {Q.}~\bibnamefont {Liu}}, \bibinfo {author} {\bibfnamefont {M.}~\bibnamefont {Gong}},\ and\ \bibinfo {author} {\bibfnamefont {X.}~\bibnamefont {Yuan}},\ }\bibfield  {title} {\bibinfo {title} {Optical vortices 30 years on: {OAM} manipulation from topological charge to multiple singularities},\ }\href {https://doi.org/10.1038/s41377-019-0194-2} {\bibfield  {journal} {\bibinfo  {journal} {Light: Science {\&} Applications}\ }\textbf {\bibinfo {volume} {8}},\ \bibinfo {pages} {90} (\bibinfo {year} {2019})}\BibitemShut {NoStop}%
\bibitem [{\citenamefont {F\"{u}rhapter}\ \emph {et~al.}(2005)\citenamefont {F\"{u}rhapter}, \citenamefont {Jesacher}, \citenamefont {Bernet},\ and\ \citenamefont {Ritsch-Marte}}]{Furhapter2005}%
  \BibitemOpen
  \bibfield  {author} {\bibinfo {author} {\bibfnamefont {S.}~\bibnamefont {F\"{u}rhapter}}, \bibinfo {author} {\bibfnamefont {A.}~\bibnamefont {Jesacher}}, \bibinfo {author} {\bibfnamefont {S.}~\bibnamefont {Bernet}},\ and\ \bibinfo {author} {\bibfnamefont {M.}~\bibnamefont {Ritsch-Marte}},\ }\bibfield  {title} {\bibinfo {title} {Spiral phase contrast imaging in microscopy},\ }\href {https://doi.org/10.1364/OPEX.13.000689} {\bibfield  {journal} {\bibinfo  {journal} {Opt. Express}\ }\textbf {\bibinfo {volume} {13}},\ \bibinfo {pages} {689} (\bibinfo {year} {2005})}\BibitemShut {NoStop}%
\bibitem [{\citenamefont {Tamburini}\ \emph {et~al.}(2006)\citenamefont {Tamburini}, \citenamefont {Anzolin}, \citenamefont {Umbriaco}, \citenamefont {Bianchini},\ and\ \citenamefont {Barbieri}}]{Tamburini2006}%
  \BibitemOpen
  \bibfield  {author} {\bibinfo {author} {\bibfnamefont {F.}~\bibnamefont {Tamburini}}, \bibinfo {author} {\bibfnamefont {G.}~\bibnamefont {Anzolin}}, \bibinfo {author} {\bibfnamefont {G.}~\bibnamefont {Umbriaco}}, \bibinfo {author} {\bibfnamefont {A.}~\bibnamefont {Bianchini}},\ and\ \bibinfo {author} {\bibfnamefont {C.}~\bibnamefont {Barbieri}},\ }\bibfield  {title} {\bibinfo {title} {Overcoming the rayleigh criterion limit with optical vortices},\ }\href {https://doi.org/10.1103/PhysRevLett.97.163903} {\bibfield  {journal} {\bibinfo  {journal} {Phys. Rev. Lett.}\ }\textbf {\bibinfo {volume} {97}},\ \bibinfo {pages} {163903} (\bibinfo {year} {2006})}\BibitemShut {NoStop}%
\bibitem [{\citenamefont {Barreiro}\ \emph {et~al.}(2008)\citenamefont {Barreiro}, \citenamefont {Wei},\ and\ \citenamefont {Kwiat}}]{Barreiro2008}%
  \BibitemOpen
  \bibfield  {author} {\bibinfo {author} {\bibfnamefont {J.~T.}\ \bibnamefont {Barreiro}}, \bibinfo {author} {\bibfnamefont {T.-C.}\ \bibnamefont {Wei}},\ and\ \bibinfo {author} {\bibfnamefont {P.~G.}\ \bibnamefont {Kwiat}},\ }\bibfield  {title} {\bibinfo {title} {{Beating the channel capacity limit for linear photonic superdense coding}},\ }\href {https://doi.org/https://doi.org/10.1038/nphys919} {\bibfield  {journal} {\bibinfo  {journal} {Nature Physics}\ }\textbf {\bibinfo {volume} {4}},\ \bibinfo {pages} {282} (\bibinfo {year} {2008})}\BibitemShut {NoStop}%
\bibitem [{\citenamefont {Grier}(2003)}]{Grier2003}%
  \BibitemOpen
  \bibfield  {author} {\bibinfo {author} {\bibfnamefont {D.~G.}\ \bibnamefont {Grier}},\ }\bibfield  {title} {\bibinfo {title} {{A revolution in optical manipulation}},\ }\href {https://doi.org/10.1038/nature01935} {\bibfield  {journal} {\bibinfo  {journal} {Nature}\ }\textbf {\bibinfo {volume} {424}},\ \bibinfo {pages} {810} (\bibinfo {year} {2003})}\BibitemShut {NoStop}%
\bibitem [{\citenamefont {{B. P. Abbott et al.}}(2016)}]{GravWave2016Blackhole}%
  \BibitemOpen
  \bibfield  {author} {\bibinfo {author} {\bibnamefont {{B. P. Abbott et al.}}} (\bibinfo {collaboration} {LIGO Scientific Collaboration and Virgo Collaboration}),\ }\bibfield  {title} {\bibinfo {title} {Observation of gravitational waves from a binary black hole merger},\ }\href {https://doi.org/10.1103/PhysRevLett.116.061102} {\bibfield  {journal} {\bibinfo  {journal} {Phys. Rev. Lett.}\ }\textbf {\bibinfo {volume} {116}},\ \bibinfo {pages} {061102} (\bibinfo {year} {2016})}\BibitemShut {NoStop}%
\bibitem [{\citenamefont {{B. P. Abbott et al.}}(2017)}]{GravWave2017NeutronStar}%
  \BibitemOpen
  \bibfield  {author} {\bibinfo {author} {\bibnamefont {{B. P. Abbott et al.}}} (\bibinfo {collaboration} {LIGO Scientific Collaboration and Virgo Collaboration}),\ }\bibfield  {title} {\bibinfo {title} {{GW170817}: Observation of gravitational waves from a binary neutron star inspiral},\ }\href {https://doi.org/10.1103/PhysRevLett.119.161101} {\bibfield  {journal} {\bibinfo  {journal} {Phys. Rev. Lett.}\ }\textbf {\bibinfo {volume} {119}},\ \bibinfo {pages} {161101} (\bibinfo {year} {2017})}\BibitemShut {NoStop}%
\bibitem [{\citenamefont {Bialynicki-Birula}\ and\ \citenamefont {Bialynicka-Birula}(2016)}]{Bialynicki-Birula_2016}%
  \BibitemOpen
  \bibfield  {author} {\bibinfo {author} {\bibfnamefont {I.}~\bibnamefont {Bialynicki-Birula}}\ and\ \bibinfo {author} {\bibfnamefont {Z.}~\bibnamefont {Bialynicka-Birula}},\ }\bibfield  {title} {\bibinfo {title} {Gravitational waves carrying orbital angular momentum},\ }\href {https://doi.org/10.1088/1367-2630/18/2/023022} {\bibfield  {journal} {\bibinfo  {journal} {New Journal of Physics}\ }\textbf {\bibinfo {volume} {18}},\ \bibinfo {pages} {023022} (\bibinfo {year} {2016})}\BibitemShut {NoStop}%
\bibitem [{\citenamefont {Baral}\ \emph {et~al.}(2020)\citenamefont {Baral}, \citenamefont {Ray}, \citenamefont {Koley},\ and\ \citenamefont {Majumdar}}]{Baral2020}%
  \BibitemOpen
  \bibfield  {author} {\bibinfo {author} {\bibfnamefont {P.}~\bibnamefont {Baral}}, \bibinfo {author} {\bibfnamefont {A.}~\bibnamefont {Ray}}, \bibinfo {author} {\bibfnamefont {R.}~\bibnamefont {Koley}},\ and\ \bibinfo {author} {\bibfnamefont {P.}~\bibnamefont {Majumdar}},\ }\bibfield  {title} {\bibinfo {title} {Gravitational waves with orbital angular momentum},\ }\href {https://doi.org/10.1140/epjc/s10052-020-7881-2} {\bibfield  {journal} {\bibinfo  {journal} {The European Physical Journal C}\ }\textbf {\bibinfo {volume} {80}},\ \bibinfo {pages} {326} (\bibinfo {year} {2020})}\BibitemShut {NoStop}%
\bibitem [{\citenamefont {Wu}\ and\ \citenamefont {Chen}(2023)}]{Wu2023}%
  \BibitemOpen
  \bibfield  {author} {\bibinfo {author} {\bibfnamefont {H.}~\bibnamefont {Wu}}\ and\ \bibinfo {author} {\bibfnamefont {L.}~\bibnamefont {Chen}},\ }\bibfield  {title} {\bibinfo {title} {Twisting and entangling gravitons in high-dimensional orbital angular momentum states via photon-graviton conversion},\ }\href {https://doi.org/10.1103/PhysRevD.107.125027} {\bibfield  {journal} {\bibinfo  {journal} {Phys. Rev. D}\ }\textbf {\bibinfo {volume} {107}},\ \bibinfo {pages} {125027} (\bibinfo {year} {2023})}\BibitemShut {NoStop}%
\bibitem [{\citenamefont {van Enk}\ and\ \citenamefont {Nienhuis}(1994)}]{VanEnk1994}%
  \BibitemOpen
  \bibfield  {author} {\bibinfo {author} {\bibfnamefont {S.~J.}\ \bibnamefont {van Enk}}\ and\ \bibinfo {author} {\bibfnamefont {G.}~\bibnamefont {Nienhuis}},\ }\bibfield  {title} {\bibinfo {title} {{Spin and orbital angular momentum of photons}},\ }\href {https://doi.org/10.1209/0295-5075/25/7/004} {\bibfield  {journal} {\bibinfo  {journal} {Europhysics Letters}\ }\textbf {\bibinfo {volume} {25}},\ \bibinfo {pages} {497} (\bibinfo {year} {1994})}\BibitemShut {NoStop}%
\bibitem [{\citenamefont {Fernandez-Corbaton}\ \emph {et~al.}(2014)\citenamefont {Fernandez-Corbaton}, \citenamefont {Zambrana-Puyalto},\ and\ \citenamefont {Molina-Terriza}}]{FernandezCorbaton2014}%
  \BibitemOpen
  \bibfield  {author} {\bibinfo {author} {\bibfnamefont {I.}~\bibnamefont {Fernandez-Corbaton}}, \bibinfo {author} {\bibfnamefont {X.}~\bibnamefont {Zambrana-Puyalto}},\ and\ \bibinfo {author} {\bibfnamefont {G.}~\bibnamefont {Molina-Terriza}},\ }\bibfield  {title} {\bibinfo {title} {On the transformations generated by the electromagnetic spin and orbital angular momentum operators},\ }\href {https://doi.org/10.1364/JOSAB.31.002136} {\bibfield  {journal} {\bibinfo  {journal} {J. Opt. Soc. Am. B}\ }\textbf {\bibinfo {volume} {31}},\ \bibinfo {pages} {2136} (\bibinfo {year} {2014})}\BibitemShut {NoStop}%
\bibitem [{\citenamefont {Palmerduca}\ and\ \citenamefont {Qin}(2024{\natexlab{a}})}]{PalmerducaQin_PT}%
  \BibitemOpen
  \bibfield  {author} {\bibinfo {author} {\bibfnamefont {E.}~\bibnamefont {Palmerduca}}\ and\ \bibinfo {author} {\bibfnamefont {H.}~\bibnamefont {Qin}},\ }\bibfield  {title} {\bibinfo {title} {Photon topology},\ }\href {https://doi.org/10.1103/PhysRevD.109.085005} {\bibfield  {journal} {\bibinfo  {journal} {Phys. Rev. D}\ }\textbf {\bibinfo {volume} {109}},\ \bibinfo {pages} {085005} (\bibinfo {year} {2024}{\natexlab{a}})}\BibitemShut {NoStop}%
\bibitem [{\citenamefont {Palmerduca}\ and\ \citenamefont {Qin}(2024{\natexlab{b}})}]{PalmerducaQin_GT}%
  \BibitemOpen
  \bibfield  {author} {\bibinfo {author} {\bibfnamefont {E.}~\bibnamefont {Palmerduca}}\ and\ \bibinfo {author} {\bibfnamefont {H.}~\bibnamefont {Qin}},\ }\bibfield  {title} {\bibinfo {title} {Graviton topology},\ }\href {https://doi.org/10.1007/JHEP11(2024)150} {\bibfield  {journal} {\bibinfo  {journal} {Journal of High Energy Physics}\ }\textbf {\bibinfo {volume} {2024}},\ \bibinfo {pages} {150} (\bibinfo {year} {2024}{\natexlab{b}})}\BibitemShut {NoStop}%
\bibitem [{\citenamefont {Leader}\ and\ \citenamefont {Lorc{\'{e}}}(2019)}]{Leader2019}%
  \BibitemOpen
  \bibfield  {author} {\bibinfo {author} {\bibfnamefont {E.}~\bibnamefont {Leader}}\ and\ \bibinfo {author} {\bibfnamefont {C.}~\bibnamefont {Lorc{\'{e}}}},\ }\bibfield  {title} {\bibinfo {title} {Corrigendum to “{The} angular momentum controversy: What's it all about and does it matter?”},\ }\href {https://doi.org/10.1016/j.physrep.2019.01.006} {\bibfield  {journal} {\bibinfo  {journal} {Physics Reports}\ }\textbf {\bibinfo {volume} {802}},\ \bibinfo {pages} {23} (\bibinfo {year} {2019})}\BibitemShut {NoStop}%
\bibitem [{\citenamefont {Shankar}(1994)}]{Shankar_QM}%
  \BibitemOpen
  \bibfield  {author} {\bibinfo {author} {\bibfnamefont {R.}~\bibnamefont {Shankar}},\ }\href {https://doi.org/https://doi.org/10.1007/978-1-4757-0576-8} {\emph {\bibinfo {title} {Principles of Quantum Mechanics}}},\ \bibinfo {edition} {2nd}\ ed.\ (\bibinfo  {publisher} {Springer},\ \bibinfo {address} {New York, NY, USA},\ \bibinfo {year} {1994})\BibitemShut {NoStop}%
\bibitem [{\citenamefont {Akhiezer}\ and\ \citenamefont {Berestetskii}(1965)}]{Akhiezer1965}%
  \BibitemOpen
  \bibfield  {author} {\bibinfo {author} {\bibfnamefont {A.}~\bibnamefont {Akhiezer}}\ and\ \bibinfo {author} {\bibfnamefont {V.}~\bibnamefont {Berestetskii}},\ }\href@noop {} {\emph {\bibinfo {title} {Quantum Electrodynamics}}},\ Interscience monographs and texts in physics and astronomy, v. 11\ (\bibinfo  {publisher} {Interscience Publishers},\ \bibinfo {address} {New York, NY, USA},\ \bibinfo {year} {1965})\BibitemShut {NoStop}%
\bibitem [{\citenamefont {Jaffe}\ and\ \citenamefont {Manohar}(1990)}]{Jaffe1990}%
  \BibitemOpen
  \bibfield  {author} {\bibinfo {author} {\bibfnamefont {R.}~\bibnamefont {Jaffe}}\ and\ \bibinfo {author} {\bibfnamefont {A.}~\bibnamefont {Manohar}},\ }\bibfield  {title} {\bibinfo {title} {The $g_1$ problem: Deep inelastic electron scattering and the spin of the proton},\ }\href {https://doi.org/https://doi.org/10.1016/0550-3213(90)90506-9} {\bibfield  {journal} {\bibinfo  {journal} {Nuclear Physics B}\ }\textbf {\bibinfo {volume} {337}},\ \bibinfo {pages} {509} (\bibinfo {year} {1990})}\BibitemShut {NoStop}%
\bibitem [{\citenamefont {Chen}\ \emph {et~al.}(2008)\citenamefont {Chen}, \citenamefont {L\"u}, \citenamefont {Sun}, \citenamefont {Wang},\ and\ \citenamefont {Goldman}}]{Chen2008}%
  \BibitemOpen
  \bibfield  {author} {\bibinfo {author} {\bibfnamefont {X.-S.}\ \bibnamefont {Chen}}, \bibinfo {author} {\bibfnamefont {X.-F.}\ \bibnamefont {L\"u}}, \bibinfo {author} {\bibfnamefont {W.-M.}\ \bibnamefont {Sun}}, \bibinfo {author} {\bibfnamefont {F.}~\bibnamefont {Wang}},\ and\ \bibinfo {author} {\bibfnamefont {T.}~\bibnamefont {Goldman}},\ }\bibfield  {title} {\bibinfo {title} {Spin and orbital angular momentum in gauge theories: Nucleon spin structure and multipole radiation revisited},\ }\href {https://doi.org/10.1103/PhysRevLett.100.232002} {\bibfield  {journal} {\bibinfo  {journal} {Phys. Rev. Lett.}\ }\textbf {\bibinfo {volume} {100}},\ \bibinfo {pages} {232002} (\bibinfo {year} {2008})}\BibitemShut {NoStop}%
\bibitem [{\citenamefont {Wakamatsu}(2010)}]{Wakamatsu2010}%
  \BibitemOpen
  \bibfield  {author} {\bibinfo {author} {\bibfnamefont {M.}~\bibnamefont {Wakamatsu}},\ }\bibfield  {title} {\bibinfo {title} {Gauge-invariant decomposition of nucleon spin},\ }\href {https://doi.org/10.1103/PhysRevD.81.114010} {\bibfield  {journal} {\bibinfo  {journal} {Phys. Rev. D}\ }\textbf {\bibinfo {volume} {81}},\ \bibinfo {pages} {114010} (\bibinfo {year} {2010})}\BibitemShut {NoStop}%
\bibitem [{\citenamefont {Bialynicki-Birula}\ and\ \citenamefont {Bialynicka-Birula}(2011)}]{Bialynicki-Birula2011}%
  \BibitemOpen
  \bibfield  {author} {\bibinfo {author} {\bibfnamefont {I.}~\bibnamefont {Bialynicki-Birula}}\ and\ \bibinfo {author} {\bibfnamefont {Z.}~\bibnamefont {Bialynicka-Birula}},\ }\bibfield  {title} {\bibinfo {title} {{Canonical separation of angular momentum of light into its orbital and spin parts}},\ }\href {https://doi.org/10.1088/2040-8978/13/6/064014} {\bibfield  {journal} {\bibinfo  {journal} {Journal of Optics}\ }\textbf {\bibinfo {volume} {13}},\ \bibinfo {pages} {064014} (\bibinfo {year} {2011})}\BibitemShut {NoStop}%
\bibitem [{\citenamefont {Leader}(2013)}]{Leader2013}%
  \BibitemOpen
  \bibfield  {author} {\bibinfo {author} {\bibfnamefont {E.}~\bibnamefont {Leader}},\ }\bibfield  {title} {\bibinfo {title} {{The angular momentum controversy: What's it all about and does it matter?}},\ }\href {https://doi.org/10.1134/S1063779613060142} {\bibfield  {journal} {\bibinfo  {journal} {Physics of Particles and Nuclei}\ }\textbf {\bibinfo {volume} {44}},\ \bibinfo {pages} {926} (\bibinfo {year} {2013})}\BibitemShut {NoStop}%
\bibitem [{\citenamefont {Leader}\ and\ \citenamefont {Lorcé}(2014)}]{Leader2014}%
  \BibitemOpen
  \bibfield  {author} {\bibinfo {author} {\bibfnamefont {E.}~\bibnamefont {Leader}}\ and\ \bibinfo {author} {\bibfnamefont {C.}~\bibnamefont {Lorcé}},\ }\bibfield  {title} {\bibinfo {title} {The angular momentum controversy: {What’s} it all about and does it matter?},\ }\href {https://doi.org/https://doi.org/10.1016/j.physrep.2014.02.010} {\bibfield  {journal} {\bibinfo  {journal} {Physics Reports}\ }\textbf {\bibinfo {volume} {541}},\ \bibinfo {pages} {163} (\bibinfo {year} {2014})}\BibitemShut {NoStop}%
\bibitem [{\citenamefont {Leader}(2016)}]{Leader2016}%
  \BibitemOpen
  \bibfield  {author} {\bibinfo {author} {\bibfnamefont {E.}~\bibnamefont {Leader}},\ }\bibfield  {title} {\bibinfo {title} {The photon angular momentum controversy: Resolution of a conflict between laser optics and particle physics},\ }\href {https://doi.org/https://doi.org/10.1016/j.physletb.2016.03.023} {\bibfield  {journal} {\bibinfo  {journal} {Physics Letters B}\ }\textbf {\bibinfo {volume} {756}},\ \bibinfo {pages} {303} (\bibinfo {year} {2016})}\BibitemShut {NoStop}%
\bibitem [{\citenamefont {Yang}\ \emph {et~al.}(2022)\citenamefont {Yang}, \citenamefont {Khosravi},\ and\ \citenamefont {Jacob}}]{Yang2022}%
  \BibitemOpen
  \bibfield  {author} {\bibinfo {author} {\bibfnamefont {L.-P.}\ \bibnamefont {Yang}}, \bibinfo {author} {\bibfnamefont {F.}~\bibnamefont {Khosravi}},\ and\ \bibinfo {author} {\bibfnamefont {Z.}~\bibnamefont {Jacob}},\ }\bibfield  {title} {\bibinfo {title} {Quantum field theory for spin operator of the photon},\ }\href {https://doi.org/10.1103/PhysRevResearch.4.023165} {\bibfield  {journal} {\bibinfo  {journal} {Phys. Rev. Res.}\ }\textbf {\bibinfo {volume} {4}},\ \bibinfo {pages} {023165} (\bibinfo {year} {2022})}\BibitemShut {NoStop}%
\bibitem [{\citenamefont {Wigner}(1939)}]{Wigner1939}%
  \BibitemOpen
  \bibfield  {author} {\bibinfo {author} {\bibfnamefont {E.}~\bibnamefont {Wigner}},\ }\bibfield  {title} {\bibinfo {title} {On unitary representations of the inhomogeneous lorentz group},\ }\href {https://doi.org/https://doi.org/10.2307/1968551} {\bibfield  {journal} {\bibinfo  {journal} {Annals of Mathematics}\ }\textbf {\bibinfo {volume} {40}},\ \bibinfo {pages} {149} (\bibinfo {year} {1939})}\BibitemShut {NoStop}%
\bibitem [{\citenamefont {Weinberg}(1995)}]{Weinberg1995}%
  \BibitemOpen
  \bibfield  {author} {\bibinfo {author} {\bibfnamefont {S.}~\bibnamefont {Weinberg}},\ }\href {https://doi.org/10.1017/CBO9781139644167} {\emph {\bibinfo {title} {The Quantum Theory of Fields}}},\ Vol.~\bibinfo {volume} {1}\ (\bibinfo  {publisher} {Cambridge University Press},\ \bibinfo {address} {Cambridge, UK},\ \bibinfo {year} {1995})\BibitemShut {NoStop}%
\bibitem [{\citenamefont {Simms}(1968)}]{Simms1968}%
  \BibitemOpen
  \bibfield  {author} {\bibinfo {author} {\bibfnamefont {D.}~\bibnamefont {Simms}},\ }\href {https://doi.org/10.1007/BFb0069914} {\emph {\bibinfo {title} {Lie Groups and Quantum Mechanics}}},\ Lecture Notes in Mathematics\ (\bibinfo  {publisher} {Springer-Verlag},\ \bibinfo {address} {Berlin, Germany},\ \bibinfo {year} {1968})\BibitemShut {NoStop}%
\bibitem [{\citenamefont {Asorey}\ \emph {et~al.}(1985)\citenamefont {Asorey}, \citenamefont {Boya},\ and\ \citenamefont {Cariñena}}]{Asorey1985}%
  \BibitemOpen
  \bibfield  {author} {\bibinfo {author} {\bibfnamefont {M.}~\bibnamefont {Asorey}}, \bibinfo {author} {\bibfnamefont {L.~J.}\ \bibnamefont {Boya}},\ and\ \bibinfo {author} {\bibfnamefont {J.~F.}\ \bibnamefont {Cariñena}},\ }\bibfield  {title} {\bibinfo {title} {Covariant representations in a fibre bundle framework},\ }\href {https://doi.org/https://doi.org/10.1016/0034-4877(85)90040-0} {\bibfield  {journal} {\bibinfo  {journal} {Reports on Mathematical Physics}\ }\textbf {\bibinfo {volume} {21}},\ \bibinfo {pages} {391} (\bibinfo {year} {1985})}\BibitemShut {NoStop}%
\bibitem [{Note1()}]{Note1}%
  \BibitemOpen
  \bibinfo {note} {Note that we use the physics convention for the Lie algebra, in which $\protect \mathfrak {so}(3)$ corresponds to unitary transformations. Mathematicians use a different convention in which $\protect \mathfrak {so}(3)$ corresponds to anti-unitary transformations. In that convention, the factors of $i$ would not appear in equations (\ref {eq:eta_g_bundle})-(\ref {eq:J_Hilbert_comm}) For a short discussion, see Section 3.4 Ref. \cite {Hall2015}.}\BibitemShut {Stop}%
\bibitem [{\citenamefont {Hall}(2013)}]{Hall2013}%
  \BibitemOpen
  \bibfield  {author} {\bibinfo {author} {\bibfnamefont {B.~C.}\ \bibnamefont {Hall}},\ }\href@noop {} {\emph {\bibinfo {title} {Quantum Theory for Mathematicians}}},\ Graduate Texts in Mathematics\ (\bibinfo  {publisher} {Springer},\ \bibinfo {address} {New York, NY, USA},\ \bibinfo {year} {2013})\BibitemShut {NoStop}%
\bibitem [{\citenamefont {Hall}(2015)}]{Hall2015}%
  \BibitemOpen
  \bibfield  {author} {\bibinfo {author} {\bibfnamefont {B.~C.}\ \bibnamefont {Hall}},\ }\href {https://doi.org/doi.org/10.1007/978-3-319-13467-3} {\emph {\bibinfo {title} {Lie Groups, Lie Algebras, and Representations: An Elementary Introduction}}},\ \bibinfo {edition} {2nd}\ ed.,\ Graduate Texts in Mathematics\ (\bibinfo  {publisher} {Springer},\ \bibinfo {address} {Cham, Switzerland},\ \bibinfo {year} {2015})\BibitemShut {NoStop}%
\bibitem [{\citenamefont {Flato}\ \emph {et~al.}(1983)\citenamefont {Flato}, \citenamefont {Sternheimer},\ and\ \citenamefont {Fronsdal}}]{Flato1983}%
  \BibitemOpen
  \bibfield  {author} {\bibinfo {author} {\bibfnamefont {M.}~\bibnamefont {Flato}}, \bibinfo {author} {\bibfnamefont {D.}~\bibnamefont {Sternheimer}},\ and\ \bibinfo {author} {\bibfnamefont {C.}~\bibnamefont {Fronsdal}},\ }\bibfield  {title} {\bibinfo {title} {{Difficulties with massless particles?}},\ }\href {https://doi.org/10.1007/BF01216186} {\bibfield  {journal} {\bibinfo  {journal} {Communications in Mathematical Physics}\ }\textbf {\bibinfo {volume} {90}},\ \bibinfo {pages} {563} (\bibinfo {year} {1983})}\BibitemShut {NoStop}%
\bibitem [{\citenamefont {Das}\ \emph {et~al.}(2024)\citenamefont {Das}, \citenamefont {Yang},\ and\ \citenamefont {Jacob}}]{Das2024}%
  \BibitemOpen
  \bibfield  {author} {\bibinfo {author} {\bibfnamefont {P.}~\bibnamefont {Das}}, \bibinfo {author} {\bibfnamefont {L.-P.}\ \bibnamefont {Yang}},\ and\ \bibinfo {author} {\bibfnamefont {Z.}~\bibnamefont {Jacob}},\ }\bibfield  {title} {\bibinfo {title} {What are the quantum commutation relations for the total angular momentum of light? tutorial},\ }\href {https://doi.org/10.1364/JOSAB.524752} {\bibfield  {journal} {\bibinfo  {journal} {J. Opt. Soc. Am. B}\ }\textbf {\bibinfo {volume} {41}},\ \bibinfo {pages} {1764} (\bibinfo {year} {2024})}\BibitemShut {NoStop}%
\bibitem [{\citenamefont {Terno}(2003)}]{Terno2003}%
  \BibitemOpen
  \bibfield  {author} {\bibinfo {author} {\bibfnamefont {D.~R.}\ \bibnamefont {Terno}},\ }\bibfield  {title} {\bibinfo {title} {Two roles of relativistic spin operators},\ }\href {https://doi.org/10.1103/PhysRevA.67.014102} {\bibfield  {journal} {\bibinfo  {journal} {Phys. Rev. A}\ }\textbf {\bibinfo {volume} {67}},\ \bibinfo {pages} {014102} (\bibinfo {year} {2003})}\BibitemShut {NoStop}%
\bibitem [{Note2()}]{Note2}%
  \BibitemOpen
  \bibinfo {note} {$s$ could also be a half integer if we were considering fermions and projective representations.}\BibitemShut {Stop}%
\bibitem [{\citenamefont {Berestetskii}\ \emph {et~al.}(1982)\citenamefont {Berestetskii}, \citenamefont {Lifshitz},\ and\ \citenamefont {Pitaevskii}}]{LandauLifshitzQED}%
  \BibitemOpen
  \bibfield  {author} {\bibinfo {author} {\bibfnamefont {V.}~\bibnamefont {Berestetskii}}, \bibinfo {author} {\bibfnamefont {E.}~\bibnamefont {Lifshitz}},\ and\ \bibinfo {author} {\bibfnamefont {L.}~\bibnamefont {Pitaevskii}},\ }\href@noop {} {\emph {\bibinfo {title} {Quantum Electrodynamics}}},\ \bibinfo {edition} {2nd}\ ed.,\ Vol.\ \bibinfo {volume} {Vol. 4}\ (\bibinfo  {publisher} {Butterworth-Heinemann},\ \bibinfo {address} {Oxford, UK},\ \bibinfo {year} {1982})\BibitemShut {NoStop}%
\bibitem [{\citenamefont {Peskin}\ and\ \citenamefont {Schroeder}(2019)}]{PeskinAndSchroeder}%
  \BibitemOpen
  \bibfield  {author} {\bibinfo {author} {\bibfnamefont {M.~E.}\ \bibnamefont {Peskin}}\ and\ \bibinfo {author} {\bibfnamefont {D.~V.}\ \bibnamefont {Schroeder}},\ }\href@noop {} {\emph {\bibinfo {title} {An Introduction to Quantum Field Theory}}}\ (\bibinfo  {publisher} {CRC Press},\ \bibinfo {address} {Boca Raton, FL, USA},\ \bibinfo {year} {2019})\BibitemShut {NoStop}%
\bibitem [{\citenamefont {Palmerduca}\ and\ \citenamefont {Qin}(2024{\natexlab{c}})}]{PalmerducaQin_helicity}%
  \BibitemOpen
  \bibfield  {author} {\bibinfo {author} {\bibfnamefont {E.}~\bibnamefont {Palmerduca}}\ and\ \bibinfo {author} {\bibfnamefont {H.}~\bibnamefont {Qin}},\ }\href {https://arxiv.org/abs/2407.03494} {\bibinfo {title} {Helicity is a topological invariant of massless particles: C=-2h}} (\bibinfo {year} {2024}{\natexlab{c}}),\ \Eprint {https://arxiv.org/abs/2407.03494} {arXiv:2407.03494 [math-ph]} \BibitemShut {NoStop}%
\bibitem [{\citenamefont {Bott}\ and\ \citenamefont {Tu}(2013)}]{Bott2013}%
  \BibitemOpen
  \bibfield  {author} {\bibinfo {author} {\bibfnamefont {R.}~\bibnamefont {Bott}}\ and\ \bibinfo {author} {\bibfnamefont {L.}~\bibnamefont {Tu}},\ }\href {https://books.google.com/books?id=COuPBAAAQBAJ} {\emph {\bibinfo {title} {Differential Forms in Algebraic Topology}}},\ Graduate Texts in Mathematics\ (\bibinfo  {publisher} {Springer},\ \bibinfo {address} {New York, NY, USA},\ \bibinfo {year} {2013})\BibitemShut {NoStop}%
\bibitem [{\citenamefont {Bott}\ and\ \citenamefont {Chern}(1965)}]{Bott1965}%
  \BibitemOpen
  \bibfield  {author} {\bibinfo {author} {\bibfnamefont {R.}~\bibnamefont {Bott}}\ and\ \bibinfo {author} {\bibfnamefont {S.~S.}\ \bibnamefont {Chern}},\ }\bibfield  {title} {\bibinfo {title} {{Hermitian vector bundles and the equidistribution of the zeroes of their holomorphic sections}},\ }\href {https://doi.org/10.1007/BF02391818} {\bibfield  {journal} {\bibinfo  {journal} {Acta Math.}\ }\textbf {\bibinfo {volume} {114}},\ \bibinfo {pages} {71} (\bibinfo {year} {1965})}\BibitemShut {NoStop}%
\bibitem [{\citenamefont {Hatcher}()}]{HatcherVBKT}%
  \BibitemOpen
  \bibfield  {author} {\bibinfo {author} {\bibfnamefont {A.}~\bibnamefont {Hatcher}},\ }\href {https://pi.math.cornell.edu/~hatcher/VBKT/VB.pdf} {\bibinfo {title} {Vector bundles and k-theory}},\ \bibinfo {howpublished} {unpublished}\BibitemShut {NoStop}%
\bibitem [{\citenamefont {Tu}(2017)}]{Tu2017differential}%
  \BibitemOpen
  \bibfield  {author} {\bibinfo {author} {\bibfnamefont {L.}~\bibnamefont {Tu}},\ }\href@noop {} {\emph {\bibinfo {title} {Differential Geometry: Connections, Curvature, and Characteristic Classes}}},\ Graduate Texts in Mathematics\ (\bibinfo  {publisher} {Springer International Publishing},\ \bibinfo {address} {Cham, Switzerland},\ \bibinfo {year} {2017})\BibitemShut {NoStop}%
\bibitem [{Note3()}]{Note3}%
  \BibitemOpen
  \bibinfo {note} {There is a two-dimensional projective UIR corresponding to spin 1/2 fermions, but this is irrelevant to the case of bosons, whose representations are non-projective.}\BibitemShut {Stop}%
\bibitem [{\citenamefont {Kohan}()}]{Khan_MSE}%
  \BibitemOpen
  \bibfield  {author} {\bibinfo {author} {\bibfnamefont {M.}~\bibnamefont {Kohan}},\ }\href {https://math.stackexchange.com/q/4886929} {\bibinfo {title} {{All actions of $\mathrm{SO}(3)$ on $S^2$ up to equivalence}}},\ \bibinfo {howpublished} {Mathematics Stack Exchange},\ \Eprint {https://arxiv.org/abs/https://math.stackexchange.com/q/4886929} {https://math.stackexchange.com/q/4886929} \BibitemShut {NoStop}%
\bibitem [{\citenamefont {Antoneli}\ \emph {et~al.}(2012)\citenamefont {Antoneli}, \citenamefont {Forger},\ and\ \citenamefont {Gaviria}}]{Antoneli2012}%
  \BibitemOpen
  \bibfield  {author} {\bibinfo {author} {\bibfnamefont {F.}~\bibnamefont {Antoneli}}, \bibinfo {author} {\bibfnamefont {M.}~\bibnamefont {Forger}},\ and\ \bibinfo {author} {\bibfnamefont {P.}~\bibnamefont {Gaviria}},\ }\bibfield  {title} {\bibinfo {title} {Maximal subgroups of compact lie groups},\ }\href {https://doi.org/10.48550/arXiv.math/0605784} {\bibfield  {journal} {\bibinfo  {journal} {Journal of Lie Theory}\ }\textbf {\bibinfo {volume} {22}},\ \bibinfo {pages} {949} (\bibinfo {year} {2012})}\BibitemShut {NoStop}%
\bibitem [{\citenamefont {Frankel}(2011)}]{Frankel2011}%
  \BibitemOpen
  \bibfield  {author} {\bibinfo {author} {\bibfnamefont {T.}~\bibnamefont {Frankel}},\ }\href@noop {} {\emph {\bibinfo {title} {The Geometry of Physics: An Introduction}}}\ (\bibinfo  {publisher} {Cambridge University Press},\ \bibinfo {address} {Cambridge, UK},\ \bibinfo {year} {2011})\BibitemShut {NoStop}%
\bibitem [{\citenamefont {Tu}(2010)}]{Tu_manifolds}%
  \BibitemOpen
  \bibfield  {author} {\bibinfo {author} {\bibfnamefont {L.~W.}\ \bibnamefont {Tu}},\ }\href {https://doi.org/10.1007/978-1-4419-7400-6} {\emph {\bibinfo {title} {An Introduction to Manifolds}}},\ Universitext\ (\bibinfo  {publisher} {Springer},\ \bibinfo {address} {New York, NY, USA},\ \bibinfo {year} {2010})\BibitemShut {NoStop}%
\bibitem [{\citenamefont {Hatcher}(2001)}]{HatcherAT}%
  \BibitemOpen
  \bibfield  {author} {\bibinfo {author} {\bibfnamefont {A.}~\bibnamefont {Hatcher}},\ }\href {https://doi.org/10.5860/choice.40-0958} {\emph {\bibinfo {title} {Algebraic Topology}}}\ (\bibinfo  {publisher} {Cambridge University Press},\ \bibinfo {address} {Cambridge, UK},\ \bibinfo {year} {2001})\BibitemShut {NoStop}%
\bibitem [{\citenamefont {Tong}(2006)}]{Tong2006}%
  \BibitemOpen
  \bibfield  {author} {\bibinfo {author} {\bibfnamefont {D.}~\bibnamefont {Tong}},\ }\href {https://www.damtp.cam.ac.uk/user/tong/qft.html} {\bibinfo {title} {Quantum field theory}},\ \bibinfo {howpublished} {unpublished} (\bibinfo {year} {2006})\BibitemShut {NoStop}%
\bibitem [{\citenamefont {Folland}(2008)}]{Folland2008}%
  \BibitemOpen
  \bibfield  {author} {\bibinfo {author} {\bibfnamefont {G.~B.}\ \bibnamefont {Folland}},\ }\href@noop {} {\emph {\bibinfo {title} {Quantum Field Theory: A Tourist Guide for Mathematicians}}},\ Mathematical surveys and monographs, vol. 149.\ (\bibinfo  {publisher} {American Mathematical Society},\ \bibinfo {address} {Providence, RI, USA},\ \bibinfo {year} {2008})\BibitemShut {NoStop}%
\end{thebibliography}%
\end{document}